%% file: main.tex
\documentclass[sigconf, nonacm]{acmart}

\usepackage{siunitx}
\usepackage{enumitem}
\usepackage{multirow}
\usepackage{pifont}
\usepackage{xcolor}
\usepackage{colortbl}
\usepackage{tcolorbox}
\usepackage{textcomp}
\usepackage{hyperref}
\usepackage[skip=1pt]{caption}
\usepackage[aboveskip=1pt]{subcaption}
\setlist[enumerate]{topsep=0pt,itemsep=0ex,partopsep=1ex,parsep=1ex}

\newcommand\vldbdoi{XX.XX/XXX.XX}
\newcommand\vldbpages{XXX-XXX}
\newcommand\vldbvolume{14}
\newcommand\vldbissue{1}
\newcommand\vldbyear{2020}
\newcommand\vldbauthors{\authors}
\newcommand\vldbtitle{\shorttitle} 
\newcommand\vldbavailabilityurl{https://github.com/qyliu-hkust/bench_search}
\newcommand\vldbpagestyle{plain} 

\begin{document}
\title{Why Are Learned Indexes So Effective but Sometimes Ineffective?}

%%
%% The "author" command and its associated commands are used to define the authors and their affiliations.
\author{Qiyu Liu}
\affiliation{%
  \institution{Southwest University}
}
\email{qyliu.cs@gmail.com}
\author{Siyuan Han}
\affiliation{
  \institution{HKUST}
}
\email{shanaj@connect.ust.hk}
\author{Yanlin Qi}
\affiliation{
  \institution{HIT Shenzhen}
}
\email{yanlinqi7@gmail.com}
\author{Jingshu Peng}
\affiliation{%
  \institution{ByteDance}
}
\email{jingshu.peng@bytedance.com}
\author{Jin Li}
\affiliation{%
  % \institution{Harvard \& MegaETH}
  \institution{Harvard University}
}
\email{jinli@g.harvard.edu}
\author{Longlong Lin}
\affiliation{%
  \institution{Southwest University}
}
\email{longlonglin@swu.edu.cn}
\author{Lei Chen}
\affiliation{%
  \institution{HKUST \& HKUST (GZ)}
}
\email{leichen@cse.ust.hk}

% \author[1]{Qiyu Liu}
% \author[2]{Siyuan Han}
% \author[2]{Jingshu Peng}
% \author[2,3]{Lei Chen}
% \affil[1]{Southwest University}
% \affil[2]{HKUST}
% \affil[3]{HKUST (GZ)}

\begin{abstract}
Learned indexes have attracted significant research interest due to their ability to offer better space-time trade-offs compared to traditional B+-tree variants. 
Among various learned indexes, the PGM-Index based on error-bounded piecewise linear approximation is an elegant data structure that has demonstrated \emph{provably} superior performance over conventional B+-tree indexes. 
In this paper, we explore two interesting research questions regarding the PGM-Index: 
\ding{182} \emph{Why are PGM-Indexes theoretically effective?} 
and \ding{183} \emph{Why do PGM-Indexes underperform in practice?} 
For question~\ding{182}, we first prove that, for a set of $N$ sorted keys, the PGM-Index can, with high probability, achieve a lookup time of $O(\log\log N)$ while using $O(N)$ space. 
To the best of our knowledge, this is the \textbf{tightest bound} for learned indexes to date. 
For question~\ding{183}, we identify that querying PGM-Indexes is highly memory-bound, where the internal error-bounded search operations often become the bottleneck. 
To fill the performance gap, we propose PGM++, a \emph{simple yet effective} extension to the original PGM-Index that employs a mixture of different search strategies, with hyper-parameters automatically tuned through a calibrated cost model. 
Extensive experiments on real workloads demonstrate that PGM++ establishes a new Pareto frontier. 
At comparable space costs, PGM++ speeds up index lookup queries by up to $\mathbf{2.31\times}$ and $\mathbf{1.56\times}$ when compared to the original PGM-Index and state-of-the-art learned indexes. 
\end{abstract}

\maketitle

\pagestyle{\vldbpagestyle}
\begingroup\small\noindent\raggedright\textbf{PVLDB Reference Format:}\\
\vldbauthors. \vldbtitle. PVLDB, \vldbvolume(\vldbissue): \vldbpages, \vldbyear.\\
\href{https://doi.org/\vldbdoi}{doi:\vldbdoi}
\endgroup
\begingroup
\renewcommand\thefootnote{}\footnote{\noindent
This work is licensed under the Creative Commons BY-NC-ND 4.0 International License. Visit \url{https://creativecommons.org/licenses/by-nc-nd/4.0/} to view a copy of this license. For any use beyond those covered by this license, obtain permission by emailing \href{mailto:info@vldb.org}{info@vldb.org}. Copyright is held by the owner/author(s). Publication rights licensed to the VLDB Endowment. \\
\raggedright Proceedings of the VLDB Endowment, Vol. \vldbvolume, No. \vldbissue\ %
ISSN 2150-8097. \\
\href{https://doi.org/\vldbdoi}{doi:\vldbdoi} \\
}\addtocounter{footnote}{-1}\endgroup

\ifdefempty{\vldbavailabilityurl}{}{
\vspace{.3cm}
\begingroup\small\noindent\raggedright\textbf{PVLDB Artifact Availability:}\\
The source code, data, and/or other artifacts have been made available at \url{\vldbavailabilityurl}.
\endgroup
}

\input{sections/0-introduction}
\input{sections/1-preliminaries}

\input{sections/2-benchmark-setting}

\input{sections/3-effective}
\input{sections/4-ineffective}
\input{sections/5-improvement}
\input{sections/6-experiment}

\input{sections/7-related-work}

\input{sections/8-conclusion}

% \begin{acks}
%  This work was supported by the [...] Research Fund of [...] (Number [...]). Additional funding was provided by [...] and [...]. We also thank [...] for contributing [...].
% \end{acks}

%\clearpage

\bibliographystyle{ACM-Reference-Format}
\bibliography{main}

\end{document}

%% file: sections/0-introduction.tex
\section{Introduction}\label{sec:introduction}
Indexes are fundamental components of DBMS and big data engines to enable real-time analytics~\cite{postgresql,spark}. 
An emerging research tendency is to directly learn the storage layout of sorted data by using simple machine learning (ML) models, leading to the concept of \emph{Learned Index}~\cite{DBLP:conf/sigmod/KraskaBCDP18, DBLP:journals/pvldb/FerraginaV20,DBLP:conf/sigmod/DingMYWDLZCGKLK20,DBLP:journals/pvldb/WuZCCWX21,DBLP:journals/pvldb/ZhangG22,zhang2024making}. 
Compared to traditional indexes like B+-tree variants~\cite{DBLP:journals/csur/Comer79,DBLP:conf/icde/LevandoskiLS13a,DBLP:conf/sigmod/KimCSSNKLBD10}, learned indexes have been shown to reduce the memory footprint by 2--3 orders of magnitude while achieving comparable index lookup performance. 

Similar to B+-trees or other binary search tree (BST) variants, learned indexes address the classical problem of \emph{Sorted Dictionary Indexing}~\cite{cormen2022introduction}. 
Given an array of $N$ sorted keys $\mathcal{K}=\{k_1,\cdots,k_N\}$, the objective of learned indexes is to find a projection function (i.e., an ML model) $f(k)\in\mathbb{N}^+$ that maps an arbitrary query key $k$ to its corresponding index in the sorted array $\mathcal{K}$ (i.e., its position on storage). 
However, ML models inherently produce prediction errors. 
As illustrated in Figure~\ref{fig:btree_vs_learned}, the maximum prediction error over $\mathcal{K}$ is denoted by $\epsilon$. 
To ensure the correctness of an index lookup query for a search key $k$, an exact ``last-mile'' search, typically a standard binary search, must be performed within the error range (i.e., $[f(k)-\epsilon, f(k)+\epsilon]$). 
% To strike a balance between the fitting capability and model complexity, instead of choosing intricate deep learning (DL) models, learned indexes like RMI~\cite{DBLP:conf/sigmod/KraskaBCDP18} or PGM-Index~\cite{DBLP:journals/pvldb/FerraginaV20} opt to stack up simple models, such as linear models or polynomial splines, in a hierarchical structure. 
To balance model accuracy with complexity, learned indexes such as Recursive Model Index (RMI)~\cite{DBLP:conf/sigmod/KraskaBCDP18} and PGM-Index~\cite{DBLP:journals/pvldb/FerraginaV20} opt to stack simple models, such as linear models or polynomial splines, in a hierarchical structure, thereby achieving a balance between the model complexity and fitting accuracy. 

\begin{figure}[t]
    \centering
    \includegraphics[width=0.37\textwidth]{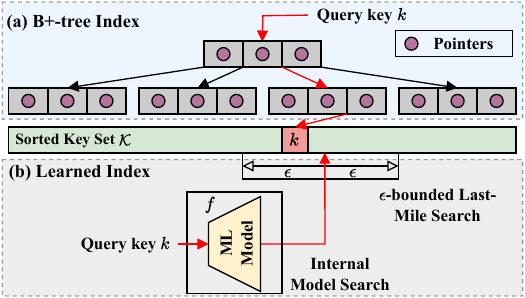}
    \caption{(a) A conventional B+-tree index. (b) A learned index with a ``last-mile'' maximum search error $\epsilon$. 
    % The position predicted by a learned index is guaranteed to be within the range $\mathsf{rank}(k)\pm\epsilon$ where $\mathsf{rank}(k)$ is the true index of the search key $k$.
    }
    \label{fig:btree_vs_learned}
    \vspace{2pt}
\end{figure}

\begin{figure*}[t]
    \centering
    \includegraphics[width=0.85\textwidth]{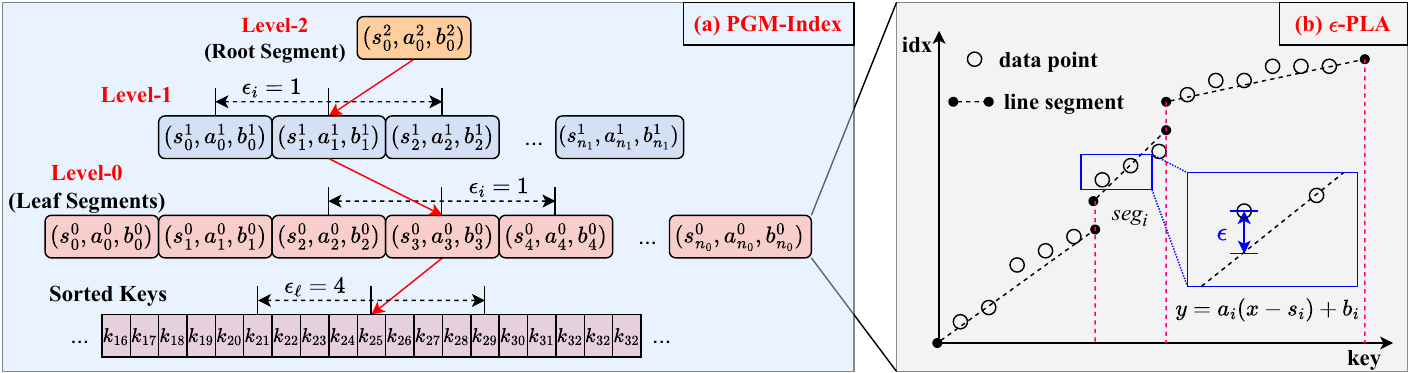}
    \caption{A toy example of a 3-level PGM-Index with $\epsilon_i=1$ (i.e., internal search error range) and $\epsilon_\ell=4$ (i.e., last-mile search error range). Processing a lookup query on such PGM-Index involves in total three linear function evaluations, two internal search operations in the range $2\cdot\epsilon_i+1$, and one ``last-mile'' search operation on the sorted data array in the range $2\cdot\epsilon_\ell+1$. }
    \label{fig:pgm}
    \vspace{-1.5ex}
\end{figure*}

Among the various published learned indexes~\cite{DBLP:conf/sigmod/KraskaBCDP18, DBLP:journals/pvldb/FerraginaV20,DBLP:conf/sigmod/DingMYWDLZCGKLK20,DBLP:journals/pvldb/WuZCCWX21,DBLP:journals/pvldb/ZhangG22,zhang2024making}, the PGM-Index~\cite{DBLP:journals/pvldb/FerraginaV20} stands out as a simple yet elegant structure that has been proven to be \emph{theoretically} more efficient than a B+-tree. 
As depicted in Figure~\ref{fig:pgm}, the PGM-Index is a multi-level structure constructed by recursively fitting \emph{error-bounded piecewise linear approximation} models ($\epsilon$-PLA). 
Searching in a PGM-Index is performed through a sequence of \emph{error-bounded search} operations in a top-down manner. 
Recent theoretical analysis~\cite{DBLP:conf/icml/FerraginaLV20} indicates that, compared to a B+-tree with fanout $B$, the PGM-Index, however, can reduce memory footprint by a factor of $B$, while preserving the same logarithmic index lookup complexity (i.e., $O(\log N)$). 

Intuitively, the PGM-Index is structured as a hierarchy of line segments, where the index height is a key factor in determining the lookup time complexity. 
Existing results~\cite{DBLP:conf/icml/FerraginaLV20,DBLP:journals/pvldb/FerraginaV20} suggest that the height of a PGM-Index built on $N$ sorted keys should be $O(\log N)$. 
However, our empirical investigations reveal that PGM-Indexes are highly \emph{flat}, with over \textbf{99\%} of the total index space cost attributed to the segments at the \emph{bottom} level. 
This observation implies that the height of the PGM-Index grows more slowly than $O(\log N)$, potentially at a sub-logarithmic rate. 
Motivated by this, we pose the following research question. 

% However, after empirically investigating the memory footprint of PGM-Index structures, we find out that PGM-Index is highly \emph{flat} where the segments in the bottom level occupy \textbf{>99\%} of the total index space cost. 
% This implies that the height growth of the PGM-Index should be even slower than $O(\log N)$, i.e., sub-logarithm. 

\noindent\textbf{{Q1: Why Are PGM-Indexes So Effective in Theory?}}
To answer this question, we establish new theoretical results for PGM-Indexes. 
With high probability (w.h.p.), the index lookup time can be bounded by $O(\log_2\log_G N)=O(\log\log N)$ using linear space of $O(N/G)$, where $G$ is a constant determined by data distribution characteristics and the error constraint $\epsilon$. 
To the best of our knowledge, this work presents the \emph{tightest} bound for learned index structures compared to existing theoretical analyses~\cite{DBLP:conf/icml/FerraginaLV20,DBLP:conf/icml/ZeighamiS23}. 

Interestingly, BSTs can be viewed as a ``materialized'' version of the binary search algorithm, whose time complexity is $O(\log N)$. 
As an analog, the PGM-Index with piecewise linear approximation models can be regarded as a ``materialized'' version of the interpolation search algorithm, whose time complexity is $O(\log\log N)$~\cite{DBLP:journals/cacm/PerlIA78,DBLP:conf/sigmod/SandtCP19}, aligning with our theoretical findings. 

Despite its theoretical superiority, recent benchmarks~\cite{DBLP:journals/pvldb/MarcusKRSMK0K20,DBLP:journals/pvldb/WongkhamLLZLW22} show that the PGM-Index falls short of practical performance expectations, often underperforming compared to well-optimized RMI variants~\cite{DBLP:conf/sigmod/KraskaBCDP18,DBLP:conf/sigmod/KipfMRSKK020}. 
This leads to our second research question.

\noindent\textbf{{Q2: Why Are PGM-Indexes Ineffective in Practice?}}
Our investigation into extensive benchmark results across various hardware platforms reveals that PGM-Indexes are memory-bound. 
% Through investigating extensive benchmark results across various hardware platforms, we find that searching on PGM-Indexes is highly \emph{memory-bounded}. 
The internal error-bounded search operation, often implemented as a standard binary search (e.g., \texttt{std::lower\_bound} in C++), becomes a bottleneck when processing an index lookup query. 
According to our benchmark (Section~\ref{sec:ineffective}), less than \textbf{1\%} of the internal segments account for over \textbf{80\%} of the total index lookup time. 

To improve search efficiency, we propose a hybrid internal search strategy that combines the advantages of linear search and highly optimized branchless binary search by properly setting search range thresholds. 
Additionally, as illustrated in Figure~\ref{fig:pgm}, constructing a PGM-Index necessitates two hyper-parameters $\epsilon_i$ and $\epsilon_\ell$, the error thresholds for internal index traversal and last-mile search on the data array, respectively. 
We find that the $\epsilon_\ell$ primarily controls the overall index size, while both $\epsilon_i$ and $\epsilon_\ell$ influence the index lookup efficiency. 
Based on theoretical analysis and experimental observations, we develop a cost model that is finely calibrated using benchmark data. 
Leveraging this cost model, we further introduce an automatic hyper-parameter tuning strategy to better balance index lookup efficiency with index size. 

In summary, our technical contributions are as follows. 
\ding{182} \textbf{New Bound.} We prove the sub-logarithmic index lookup time of the PGM-Index (i.e., $O(\log\log N)$). 
This result tightens the previous logarithmic bound on the PGM-Index and further validates its provable performance superiority compared to conventional tree-based indexes. 
\ding{183} \textbf{Simple Methods.} We introduce PGM++, a \emph{simple yet effective} improvement to the PGM-Index by replacing the costly internal search operations. 
We further devise an automatic parameter tuner for PGM++, guided by an accurate cost model. 
\ding{184} \textbf{New Pareto Frontier.} Extensive experimental studies on real and synthetic data show that, with a comparable index memory footprint, PGM++ robustly outperforms the original PGM-Index and optimized RMI variants~\cite{DBLP:journals/pvldb/MarcusKRSMK0K20,DBLP:journals/pvldb/ZhangG22} by up to \textbf{2.31\texttimes} and \textbf{1.56\texttimes}, respectively. 
For example, even on a resource-constrained device like MacBook Air 2024~\cite{macbook2024}, PGM++ achieves index lookup time of \textbf{<400 ns} on \textbf{800 million} keys, using only \textbf{0.28 MB} of memory. 

The remainder of this paper is structured as follows. 
Section~\ref{sec:preliminaries} introduces the basis of learned indexes, followed by the micro-benchmark setup details in Section~\ref{sec:benchmark_setting}. 
% We introduce the basis of learned indexes and the micro-benchmark setup details in Section~\ref{sec:preliminaries} and Section~\ref{sec:benchmark_setting}, respectively. 
Section~\ref{sec:theory} presents our core theoretical analysis of the PMG-Index. 
Section~\ref{sec:ineffective} explores the reasons behind the PGM-Index's underperformance in practice. 
In Section~\ref{sec:optimization}, we introduce PGM++, an optimized PGM-Index variant featuring hybrid error-bounded search and automatic hyper-parameter tuning. 
Section~\ref{sec:exp} reports the experimental results. 
Section~\ref{sec:related_works} surveys and discusses related works, and finally, Section~\ref{sec:conclusion} concludes the paper and discusses future studies. 

%% file: sections/1-preliminaries.tex
\section{Preliminaries}\label{sec:preliminaries}
We first overview the basis of learned indexes ($\vartriangleright$ Section~\ref{subsec:pgm}) and then elaborate on the details of existing theoretical results ($\vartriangleright$ Section~\ref{subsec:exisiting_theory}). 
Table~\ref{tab:notations} summarizes the major notations. 

\begin{table}[t]
    \centering
    \caption{Summary of major notations.}
    \label{tab:notations}
    \small
    \begin{tabular}{cc}
    \toprule
        \cellcolor[HTML]{f4f5f6}\textbf{Notation} & \cellcolor[HTML]{f4f5f6}\textbf{Description} \\\midrule
        $\mathcal{K}$ & a set of $N$ sorted keys \\
        $\mathsf{rank}(k)$ & the sorting index of a key $k$ in $\mathcal{K}$\\
        $\epsilon_i$ & the internal search error parameter of PGM-Index\\
        $\epsilon_\ell$ & the last-mile search error parameter of PGM-Index\\
        $g_i$ & the difference between $k_i$ and $k_{i-1}$ (a.k.a.~gap)\\
        $\mu,\sigma^2$ & the mean and variance of gap distribution\\
        $(s,a,b)$ & a line segment $\ell(x)=a\cdot(x-s)+b$\\
        $H_{PGM}$ & the height of a PGM-Index\\
    \bottomrule
    \end{tabular}
\end{table}

\vspace{-1ex}
\subsection{Learned Index}\label{subsec:pgm}
Given a set of $N$ \emph{sorted} keys $\mathcal{K}=\{k_1,k_2,\cdots,k_N\}$ and an index set $\mathcal{I}=\{1,2,\cdots,N\}$, 
the goal of learned indexes is to find a mapping function $f(k)\in\mathbb{N}^+$ such that $f$ can project a search key $k\in\mathcal{K}$ to its corresponding index $\mathsf{rank}(k)\in\mathcal{I}$ with controllable error. 
Intuitively, learning $f$ is equivalent to learning a cumulative distribution function (CDF) scaled by the data size $N$. 
The model selection considerations for $f$ are threefold: 

\noindent\ding{182} \textbf{Compactness}: the model $f$ should be compact to reduce memory footprint, and model inference using $f$ must not introduce significant computational overhead;

\noindent\ding{183} \textbf{Error-Boundness}: the model $f$ should be error-bounded, ensuring that an exact last-mile search can correct prediction errors, i.e., $|f(k)-\mathsf{rank}(k)|\leq\epsilon$ for $\forall k\in\mathcal{K}$; 

\noindent\ding{184} \textbf{Monotonicity}: to ensure the correctness of querying keys outside $\mathcal{K}$, $f(k_1)\leq f(k_2)$ should hold for any $k_1\leq k_2$. 

Since running deep learning (DL) models usually require a heavy runtime like PyTorch~\cite{pytorch} or TensorFlow~\cite{tensorflow} that are costly and less flexible, existing learned index designs favor \emph{stacking} simple models, such as linear functions~\cite{DBLP:conf/sigmod/DingMYWDLZCGKLK20,DBLP:journals/pvldb/WuZCCWX21,DBLP:journals/pvldb/FerraginaV20}, polynomial splines~\cite{DBLP:conf/sigmod/KraskaBCDP18}, and radix splines~\cite{DBLP:conf/sigmod/KipfMRSKK020}. 
Among these learned index structures, the PGM-Index~\cite{DBLP:journals/pvldb/FerraginaV20} employs the error-bounded piecewise linear approximation ($\epsilon$-PLA) to strike a balance between the model complexity and prediction accuracy, which is defined as follows. 

\begin{definition}[$\epsilon$-PLA]
     Given a univariate set $\mathcal{X}=\{x_1,\cdots, x_N\}$, a corresponding target set $\mathcal{Y}=\{y_1,\cdots,y_N\}$, and an error constraint $\epsilon$,
     an $\epsilon$-PLA on the point set in Cartesian space $(\mathcal{X}, \mathcal{Y})=\{(x_i,y_i)\}_{i=1,\cdots,N}$ is defined as,
     \begin{equation}\label{eq:pla}
         f(x)=\begin{cases}
             a_1\cdot (x-s_1)+b_1&\text{ if } s_1\leq x< s_2\\
             a_2\cdot (x-s_2)+b_2&\text{ if } s_2\leq x<s_3\\
             \quad\cdots&\qquad\cdots\\
             a_{m}\cdot (x-s_m)+b_{m}&\text{ if } s_{m}\leq x<+\infty\\
         \end{cases}
     \end{equation}
     such that for $\forall i=1,2,\cdots,N$, it always holds that $|f(x_i)-y_i|\leq\epsilon$. 
\end{definition}

% Figure~\ref{fig:pgm}(b) exemplifies an $\epsilon$-PLA. 
The $i$-th segment in Eq.~\eqref{eq:pla} can be expressed by a tuple $seg_i=(s_i, a_i, b_i)$ where $s_i$ is the segment starting point, $a_i$ is the slope, and $b_i$ is the intercept. 
To ensure the monotonic requirement, the segments in Eq.~\eqref{eq:pla} should satisfy two conditions: (a) $a_i\geq 0$ for $i=1,\cdots,m$, and (b) $s_i<s_j$ for $\forall 1\leq i<j\leq m$. 
We then extend the original PGM-Index definition~\cite{DBLP:journals/pvldb/FerraginaV20} by separating the error parameters for internal search and last-mile search. 

\begin{definition}[$(\epsilon_i, \epsilon_\ell)$-PGM-Index~\cite{DBLP:journals/pvldb/FerraginaV20}]\label{def:pgm_construction}
    Given a sorted key set $\mathcal{K}=\{k_1,k_2,\cdots,k_N\}$ and two error parameters $\epsilon_i$ and $\epsilon_\ell$ ($\epsilon_i, \epsilon_\ell\in\mathbb{N}^+$), an $(\epsilon_i, \epsilon_\ell)$-PGM-Index is a multi-level structure where the bottom level (a.k.a., the leaf level or level-0) is an $\epsilon_\ell$-PLA and the remaining levels (a.k.a., internal levels) are $\epsilon_i$-PLA(s). 
    The structure can be constructed in a \emph{bottom-up} manner:
    
    \noindent\ding{182} \textbf{Leaf Level}: an $\epsilon_\ell$-PLA constructed on $(\mathcal{K}, \mathcal{I}=\{1,\cdots,N\})$. 

    \noindent\ding{183} \textbf{Internal Levels}: for the $j$-th level ($j\geq 1$), let $\mathcal{S}_{j-1}$ denote the set of segments in the $(j-1)$-th level (i.e., the previous level), and let $\mathcal{K}_{j-1}=\{seg.s\mid seg \in \mathcal{S}_{j-1}\}$ and $\mathcal{I}_{j-1}=\{1,2,\cdots,|\mathcal{K}_{j-1}|\}$. 
    Then, the $j$-th level is an $\epsilon_i$-PLA constructed on dataset $(\mathcal{K}_{j-1}, \mathcal{I}_{j-1})$. 

    \noindent\ding{184} \textbf{Root Level}: the internal level consisting of a \emph{single} line segment. 
\end{definition}

Specifically, we denote the segments in the bottom level as \textit{\underline{l}eaf segments} and the remaining segments as \textit{\underline{i}nternal segments} (corresponding to the subscripts of $\epsilon_\ell$ and $\epsilon_i$, respectively). 
The following example illustrates the PGM-Index lookup query processing.

\begin{example}[PGM-Index Lookup]\label{eg:pgm_query}
    Figure~\ref{fig:pgm} illustrates a 3-level PGM-Index with $\epsilon_i=1$ and $\epsilon_\ell=4$. 
    Given a query key $k$, an index lookup query is performed in a \emph{top-down} manner from the root level to the bottom level as follows:
    
    \noindent\ding{182} The \textbf{Internal Index Traversal} phase starts from the root level and finds the appropriate line segment in each level until reaching the bottom level (depicted by the red path in Figure~\ref{fig:pgm}). 
    Specifically, let $seg^{j}=(s^{j},a^{j},b^{j})$ denote the segment in the $j$-th level during the traversal. 
    The next segment in the $(j-1)$-th level to be traversed is found by searching $k$ within range $a^{j}\cdot(k-s^{j})+b^{j}\pm\epsilon_i$. 
    
    \noindent\ding{183} The \textbf{Last-Mile Search} phase performs an exact search on the raw sorted keys (i.e., $\mathcal{K}$) within the range $\widehat{\mathsf{rank}(k)}\pm\epsilon_\ell$ where $\widehat{\mathsf{rank}(k)}=a^{0}\cdot(k-s^{0})+b^{0}$ is the predicted rank and $seg^{0}=(s^{0},a^{0},b^{0})$ is the \emph{leaf segment} found during the internal index traversal phase. 
\end{example}

Recall that the index construction procedures introduced in Definition~\ref{def:pgm_construction} guarantees that the maximum errors for internal index traversal and last-mile search cannot exceed $\epsilon_i$ and $\epsilon_\ell$, respectively. 
Thus, the aforementioned lookup processing ensures the correct location (i.e., $\mathsf{rank}(k)$) must be found for an arbitrary query key $k$. 

\begin{table}[t]
    \centering
    \caption{Summary of theoretical results. For our result, $G$ is a constant that depends on data distribution characteristics and the pre-specified error bound $\epsilon$. }
    \label{tab:theory}
    \small
    \begin{tabular}{cccc}
    \toprule
        \cellcolor[HTML]{f4f5f6}\textbf{Results} & \cellcolor[HTML]{f4f5f6}\textbf{Base Model} & \cellcolor[HTML]{f4f5f6}\textbf{Lookup Time} & \cellcolor[HTML]{f4f5f6}\textbf{Space Cost} \\\midrule
        ICML'20~\cite{DBLP:conf/icml/FerraginaLV20} & Linear & $O(\log N)$ & $O(N/\epsilon^2)$ \\
        ICML'23~\cite{DBLP:conf/icml/ZeighamiS23} & Constant & $O(\log\log N)$ & $O(N\log N)$ \\
        \cellcolor[HTML]{B4FFB4}\textbf{Ours} & \cellcolor[HTML]{B4FFB4}Linear & \cellcolor[HTML]{B4FFB4}$O(\log\log N)$ & \cellcolor[HTML]{B4FFB4}$O(N/G)$ \\
    \bottomrule
    \end{tabular}
\end{table}

\begin{table*}
\caption{Summary of three micro-benchmark platforms. For platforms \texttt{X86-1} and \texttt{ARM} whose CPU chips adopt the ``big.LITTLE'' architecture~\cite{biglittle}, the hardware statistics of the performance cores (i.e., P-core) are reported. 
The reported L1/L2/L3 sizes represent the actual cache size that a physical core can access. 
Notably, for the Apple M3 chip, only L1 cache and L2 cache are available. }\label{tab:platforms}
\small
\begin{tabular}{ccccccccc}
\toprule
\cellcolor[HTML]{f4f5f6}\textbf{Platform}  & \cellcolor[HTML]{f4f5f6}\textbf{OS} & \cellcolor[HTML]{f4f5f6}\textbf{Compiler}           & \cellcolor[HTML]{f4f5f6}\textbf{CPU}                  & \cellcolor[HTML]{f4f5f6}\textbf{Frequency} & \cellcolor[HTML]{f4f5f6}\textbf{Memory}    & \cellcolor[HTML]{f4f5f6}\textbf{L1}     & \cellcolor[HTML]{f4f5f6}\textbf{L2}     & \cellcolor[HTML]{f4f5f6}\textbf{L3 (LLC)} \\
\midrule
\texttt{X86-1} & Ubuntu 20.04 & g++ 11 & Intel Core i7-13700K & 5.30 GHz (P-core)  & 32 GB DDR4 & 64 KiB & 256 KiB  & 16 MB    \\
\texttt{X86-2} & CentOS 9.4 & g++ 11 & AMD EPYC 7413 & 3.60 GHz   & 1 TB DDR4 & 64 KiB & 1 MB & 256 MB   \\
\texttt{ARM} & macOS 14.4.1 & clang++ 15 & Apple M3             & 4.05 GHz (P-core)          & 16 GB LPDDR5    & 320 KiB       & 16 MB   & N.A.        \\
\bottomrule
\end{tabular}
\vspace{-1.5ex}
\end{table*}

\subsection{Existing Theoretical Results}\label{subsec:exisiting_theory}
From Section~\ref{subsec:pgm}, two key questions need to be addressed to determine the space and time complexities of the PGM-Index. 
\ding{182} \emph{How many line segments are required to satisfy the error constraint for an $\epsilon$-PLA model?} 
and \ding{183} \emph{What is the height (i.e., the number of layers) of a PGM-Index?} 
In this section, we review the related theoretical studies~\cite{DBLP:conf/icml/FerraginaLV20,DBLP:conf/icml/ZeighamiS23} regarding these two questions, with major results summarized in Table~\ref{tab:theory}. 

The original PGM-Index~\cite{DBLP:journals/pvldb/FerraginaV20} first provides a straightforward lower bound to determine the index height. 

\begin{theorem}[PGM-Index Lower Bound~\cite{DBLP:journals/pvldb/FerraginaV20}]\label{theorem:pgm_height}
    Given a consecutive chunk of $2\epsilon+1$ sorted keys $\{k_i,\cdots,k_{i+2\epsilon}\}\subseteq\mathcal{K}$, there exists a horizontal line segment $\ell(x)=i+\epsilon$ such that $|\ell(k_j)-j|\leq\epsilon$ holds for $j=i,\cdots,i+2\epsilon$, implying that each line segment in an $\epsilon$-PLA can cover at least $2\epsilon+1$ keys. 
\end{theorem}

Recall the recursive construction process in Definition~\ref{def:pgm_construction}, w.l.o.g., a PGM-Index with $\epsilon_i=\epsilon_\ell=\epsilon$ has a height of $O(\log_\epsilon N)=O(\log N)$. 
Thus, the index lookup takes time $O(\log N\cdot\log_2\epsilon)=O(\log N)$ as $\epsilon$ can be regarded as a pre-specified constant. 

Ferragina~et~al.~\cite{DBLP:conf/icml/FerraginaLV20} further tighten the results in Theorem~\ref{theorem:pgm_height} by showing that the expected segment coverage is proportional to $\epsilon^2$. 
Suppose that the key set to be indexed $\mathcal{K}=\{k_1,k_2,\cdots,k_N\}$ is a materialization of a random process $k_i = k_{i-1} + g_i$ for $i\geq 2$ where $g_i$'s are i.i.d.~random variables (r.v.) following some unknown distribution. 
We term the r.v.~$g_i$ as the ``gap'' and denote $\mu=\mathbf{E}[g_i]$ and $\sigma^2=\mathbf{Var}[g_i]$ as its mean and variance.
These distribution characteristics are crucial in determining the expected number of segments to satisfy the error constraints.

\begin{theorem}[Expected Line Segment Coverage~\cite{DBLP:conf/icml/FerraginaLV20}]\label{theorem:segment_coverage}
    Given a set of sorted keys $\mathcal{K}=\{k_1,k_2,\cdots,k_N\}$ and an error parameter $\epsilon$, let the gap be $g_i=k_i-k_{i-1}$. 
    If the condition $\epsilon\gg\sigma/\mu$ holds, with high probability, the expected number of keys in $\mathcal{K}$ covered by a line segment $\ell(x)=\mu\cdot (x- k_1)+1$ is given by 
    \begin{equation}\label{eq:pgm_coverage}
        \mathbf{E}\left[\min\left\{i\in \mathbb{N}^+\mid |\ell(k_i)-i|> \epsilon\right\}\right]={\mu^2\epsilon^2}/{\sigma^2}, 
    \end{equation}
    where $\ell(k_i)=\mu\cdot(k_i-k_1)+1$ is the predicted index for a key $k_i$. 
\end{theorem}

By constructing a special line segment with slope $\mu$, Theorem~\ref{theorem:segment_coverage} establishes the relationship between the expected segment coverage and the error constraint $\epsilon$. 
Based on Theorem~\ref{theorem:segment_coverage}, for a set of $N$ sorted keys, the expected number of segments\footnote{The conclusion is drawn hastily as, in general, $1/\mathbf{E}[X]\neq\mathbf{E}[1/X]$ for an arbitrary random variable $X$. 
A more rigorous proof can be found in Theorem~4 of~\cite{DBLP:conf/icml/FerraginaLV20}.} of a \emph{one-layer} $\epsilon$-PLA can be derived as ${N\sigma^2}/{\epsilon^2\mu^2}$. 
In the practical PGM-Index implementation, an \emph{optimal} $\epsilon$-PLA fitting algorithm~\cite{DBLP:journals/cacm/ORourke81} is adopted to \emph{minimize} the number of line segments while ensuring the error constraint $\epsilon$ is met. 
Thus, the expected number of segments can be then bounded by $O({N\sigma^2}/{\epsilon^2\mu^2})$.

Combining the results in Theorem~\ref{theorem:pgm_height} and Theorem~\ref{theorem:segment_coverage}, Ferragina~et~al.~\cite{DBLP:conf/icml/FerraginaLV20} conclude that a PGM-Index with $\epsilon_i=\epsilon_\ell=\epsilon$ using $O(N/\epsilon^2)$ space can handle lookup queries in $O(\log N)$ time with high probability. 
By setting $\epsilon=\Theta(B)$, a PGM-Index can achieve the same logarithmic index lookup complexity of a B+-tree while reducing the space complexity from B+-tree's $O(N/B)$ to $O(N/B^2)$. 

In addition to~\cite{DBLP:conf/icml/FerraginaLV20}, a recent study~\cite{DBLP:conf/icml/ZeighamiS23} also delves into the theoretical aspects of learned index. 
They demonstrate that a Recursive Model Index~\cite{DBLP:conf/sigmod/KraskaBCDP18} using \emph{piece-wise constant} functions as base models can achieve a sub-logarithmic lookup complexity of $O(\log\log N)$ at the cost of \emph{super-linear} space, specifically $O(N\log N)$. 

\noindent\textbf{Our Results.} Inspired by the findings in~\cite{DBLP:conf/icml/ZeighamiS23}, we reasonably speculate that PGM-Indexes, utilizing $\epsilon$-PLA as base models, can achieve the same \emph{sub-logarithmic} lookup time complexity with \emph{reduced} space overhead, given that a constant function can be regarded as a special case of a piecewise linear function. 
As summarized in Table~\ref{tab:theory}, our analysis in Section~\ref{sec:theory} concludes that, w.h.p., the PGM-Index can search a query key in $O(\log\log N)$ time while requiring only linear space $O(N/G)$, where $G$ is a constant related to the error parameter $\epsilon$ and gap distribution characteristics.

%% file: sections/2-benchmark-setting.tex
\begin{table}[t]
    \centering
    \caption{Statistics of benchmark datasets. $h_D$ is the distribution hardness ratio. $\overline{Cov}$ is the observed segment coverage to fit a PLA model with an error bound of $\epsilon=16$.}\label{tab:datasets}
    \small
    \begin{tabular}{cccccc}
    \toprule
        \cellcolor[HTML]{f4f5f6}\textbf{Dataset} & \cellcolor[HTML]{f4f5f6}\textbf{Category} & \cellcolor[HTML]{f4f5f6}\textbf{\#Keys} & \cellcolor[HTML]{f4f5f6}\textbf{Raw Size} & \cellcolor[HTML]{f4f5f6}$h_{D}$ & \cellcolor[HTML]{f4f5f6}$\overline{Cov}$\\\midrule
         \texttt{fb} & Real & 200 M & 1.6 GB & 3.88 & 94 \\
         \texttt{wiki} & Real & 200 M & 1.6 GB & 1.77 & 877\\
         \texttt{books} & Real & 800 M & 6.4 GB & 5.39 & 101\\
         \texttt{osm} & Real & 800 M & 6.4 GB & 1.91 & 129\\
         % \midrule
         % \texttt{uniform} & Synthetic & 400 M & 3.2 GB & Varied & N.A. \\ 
         % \texttt{normal} & Synthetic & 400 M & 3.2 GB & Varied & N.A. \\ 
         % \texttt{lognormal} & Synthetic & 400 M & 3.2 GB & Varied & N.A. \\
         \bottomrule
    \end{tabular}
\end{table}

\section{Microbenchmark Setting}\label{sec:benchmark_setting}
To ensure consistency in presentation, this section outlines the microbenchmark setups, including the hardware platforms, datasets, and query workloads. 
The remainder of this paper adopts this microbenchmark to either motivate or validate the theoretical findings and proposed methodologies. 

\noindent\textbf{Platforms.} We perform the subsequent experiments on three platforms with different architectures: 
\ding{182} \texttt{X86-1} is an Ubuntu desktop equipped with an Intel\textcopyright~Core\texttrademark~i7-13700K CPU (5.30 GHz, P-core) and 64 GB of memory; 
\ding{183} \texttt{X86-2} is a CentOS server with 2 AMD\textcopyright~EPYC\texttrademark~7413 CPUs (3.60 GHz) and 1 TB of memory; 
and \ding{184} \texttt{ARM} is a Macbook Air laptop with an Apple Silicon M3 CPU (4.05 GHz, P-core) and 16 GB of unified memory, which offers higher memory bandwidth compared to the \texttt{X86} platforms. 
As we will discuss in Section~\ref{sec:ineffective}, searching a PGM-Index is highly memory-bound, and factors such as cache latency and memory bandwidth can significantly affect query performance\footnote{Typical access latencies for L1 cache, last level cache (LLC), and main memory are 1 ns, 20 ns, and 100 ns, respectively.}. 
Table~\ref{tab:platforms} summarizes the specifications of the benchmark platforms. 

In addition, all the experiments are written in C++ and compiled using g++ 11.4 on \texttt{X86-1} and \texttt{X86-2} and clang++ 15 on \texttt{ARM}. 
The complete microbenchmark implementation and experimental results are publicly available at~\cite{pgm++}.

\noindent\textbf{Benchmark Datasets.} 
We adopt 4 real datasets from SOSD~\cite{DBLP:journals/pvldb/MarcusKRSMK0K20} that have been widely evaluated in previous studies~\cite{DBLP:conf/sigmod/DingMYWDLZCGKLK20,DBLP:conf/sigmod/KipfMRSKK020,DBLP:journals/pvldb/WongkhamLLZLW22,DBLP:journals/pvldb/ZhangG22,DBLP:journals/pvldb/WuZCCWX21}. 
Specifically, \ding{182} \texttt{fb} is a set of user IDs randomly sampled from Facebook~\cite{DBLP:conf/sigmod/SandtCP19}; 
\ding{183} \texttt{wiki} is a set of edit timestamp IDs committed to Wikipedia~\cite{wikidata}; 
\ding{184} \texttt{books} is the dataset of book popularity from Amazon; and 
\ding{185} \texttt{osm} is a set of cell IDs from OpenStreetMap~\cite{openstreetmap}. 
We also generate 3 synthetic datasets by sampling from uniform, normal, and log-normal distributions, following a process similar to~\cite{DBLP:journals/pvldb/MarcusKRSMK0K20,zhang2024making}. 
All keys are stored as 64-bit unsigned integers (\texttt{uint64\_t} in C++), and Table~\ref{tab:datasets} summarizes the dataset statistics. 

To quantify the difficulty of indexing a dataset, we define the \emph{distribution hardness ratio} as $h_D=\sigma^2/\mu^2$ where $\mu$ and $\sigma^2$ represent the mean and variance of the gap distribution for a dataset. 
According to Theorem~\ref{theorem:segment_coverage}, a higher $h_D$ implies a harder dataset to learn, as more segments are required to meet the error constraint $\epsilon$. 
However, as illustrated in Figure~\ref{fig:gap_dist}, extreme values can easily influence $h_D$, leading to an overestimation of the necessary segment count. 
For example, on dataset \texttt{osm}, the original hardness ratio $h_D=1.27\times 10^6$, and according to Theorem~\ref{theorem:segment_coverage}, the expected segment coverage for an $\epsilon$-PLA with $\epsilon=16$ can be estimated as, 
\begin{equation}
\small
    \frac{\mu^2\cdot\epsilon^2}{\sigma^2} = \frac{\epsilon^2}{h_D} = \frac{16^2}{1.27\times 10^6} \approx 0.0002,
\end{equation}
which is far away from 129, the observed segment coverage. 
 % (column $\overline{Cov}$ in Table~\ref{tab:datasets})

To mitigate this, we clip the observed gaps at the 1\%- and 99\%-quantiles and re-calculate $h_D$ based on the clipped gaps (as reported in column $h_D$ of Table~\ref{tab:datasets}). 
After removing the extreme gaps, the revised $h_D=5.39$ on dataset \texttt{osm}, and the corresponding estimated segment coverage is $16^2/5.39\approx 71.3$, which is much closer to the observed value. 
Notably, accurately estimating the segment coverage is vital to building an effective cost model. 
The estimator based on clipped gaps remains too \emph{coarse}, as it fails to capture local data variations. 
In Section~\ref{subsec:cost_model}, we will develop a more fine-grained coverage estimator based on adaptive data partition. 

\noindent\textbf{Query Workloads.} 
Similar to the settings in~\cite{DBLP:conf/sigmod/KraskaBCDP18,DBLP:journals/pvldb/FerraginaV20,DBLP:conf/sigmod/KipfMRSKK020,DBLP:journals/pvldb/ZhangG22}, our work focuses on \emph{in-memory read-heavy} workloads. 
Given a key set $\mathcal{K}$, we generate the query workload by randomly sampling $S$ (by default $S=5,000$) keys from $\mathcal{K}$. 
To simulate different access patterns, we sample lookup keys from two distributions: 
\ding{182} \emph{Uniform}, where every key in $\mathcal{K}$ has an equal likelihood of being sampled; 
and \ding{183} \emph{Zipfan}, where the probability of sampling the $i$-th key in $\mathcal{K}$ is given by $p(i)=i^\alpha/\sum_{j=1}^{N}j^\alpha$. 
For the Zipfan workload, by default, we set the parameter $\alpha=1.3$ such that over 90\% of index accesses occur within the range of $(0, 10^3]$.

\begin{figure}[t]
    \centering
    \includegraphics[width=0.36\textwidth]{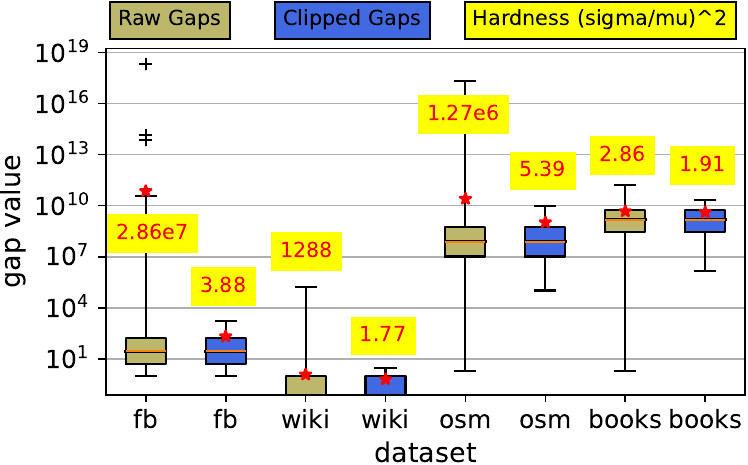}
    \caption{Gap distributions for 4 real datasets. In the box plots, the horizontal lines and the star marks refer to the medians and means of data, respectively.}
    \label{fig:gap_dist}
\end{figure}

%% file: sections/3-effective.tex
\begin{table}[t]
    \centering
    \caption{Statistics of PGM-Index ($\epsilon_i=\epsilon_\ell=\epsilon$) and B+-tree (fan-out $B=\epsilon$) under different configurations. 
    $\epsilon$ is fixed to 8 when varying data size $N$ (synthetic uniform keys), and $N$ is fixed to 800M when varying $\epsilon$ (real dataset \texttt{books}).
    The ratio in percentage refers to the proportion of leaf segments contributing to the total index memory footprint.}
    \label{tab:index_height_experiment}
    \small
    \begin{tabular}{c|ccccc}
    \toprule
        \cellcolor[HTML]{f4f5f6}$N$ & \multicolumn{1}{c}{\cellcolor[HTML]{f4f5f6}\textbf{\begin{tabular}[c]{@{}c@{}}\textbf{PGM}\\ \textbf{Height}\end{tabular}}} & \multicolumn{1}{c}{\cellcolor[HTML]{f4f5f6}\begin{tabular}[c]{@{}c@{}}\textbf{Leaf}\\ \textbf{Segments}\end{tabular}} & \multicolumn{1}{c}{\cellcolor[HTML]{f4f5f6}\begin{tabular}[c]{@{}c@{}}\textbf{Internal}\\ \textbf{Segments}\end{tabular}} & \multicolumn{1}{c}{\cellcolor[HTML]{f4f5f6}\begin{tabular}[c]{@{}c@{}}\textbf{\% over}\\\textbf{Total}\end{tabular}} & \multicolumn{1}{c}{\cellcolor[HTML]{f4f5f6}\textbf{\begin{tabular}[c]{@{}c@{}}B+-tree\\ Height\end{tabular}}} \\\midrule
        $10^3$ & 2 & 2 & 2 & 50.0\% & 4 \\
        $10^4$ & 2 & 16 & 2 & 88.9\% & 6\\
        $10^5$ & 2 & 140 & 2 & 98.6\% & 7\\
        $10^6$ & 2 & 1,388 & 2 & 99.9\% & 8\\
        $10^7$ & 3 & 13,918 & 12 & 99.9\% & 9\\
        $10^8$ & 3 & 139,376 & 109 & 99.9\% & 10\\
        $10^9$ & 4 & 1,394,003 & 1,049 & 99.9\% & 11\\  
        \midrule
        \cellcolor[HTML]{f4f5f6}$\epsilon$ ($B$) & \multicolumn{1}{c}{\cellcolor[HTML]{f4f5f6}\textbf{\begin{tabular}[c]{@{}c@{}}\textbf{PGM}\\ \textbf{Height}\end{tabular}}} & \multicolumn{1}{c}{\cellcolor[HTML]{f4f5f6}\begin{tabular}[c]{@{}c@{}}\textbf{Leaf}\\ \textbf{Segments}\end{tabular}} & \multicolumn{1}{c}{\cellcolor[HTML]{f4f5f6}\begin{tabular}[c]{@{}c@{}}\textbf{Internal}\\ \textbf{Segments}\end{tabular}} & \multicolumn{1}{c}{\cellcolor[HTML]{f4f5f6}\begin{tabular}[c]{@{}c@{}}\textbf{\% over}\\\textbf{Total}\end{tabular}} & \multicolumn{1}{c}{\cellcolor[HTML]{f4f5f6}\textbf{\begin{tabular}[c]{@{}c@{}}B+-tree\\ Height\end{tabular}}} \\\midrule
        % 4 & \cellcolor[HTML]{B4FFB4}5 & 31,502,379 & 409,962 & 98.7\% & \cellcolor[HTML]{FFECEC}16 \\
        8 & 4 & 16,859,902 & 46,572 & 99.7\% & 11\\
        16 & 4 & 7,943,403 & 4,100 & 99.9\% & 9\\
        32 & 3 & 2,464,229 & 272 & 99.99\% & 7\\
        64 & 3 & 797,152 & 60 & 99.99\% & 6\\
        128 & 3 & 267,966 & 25 & 99.99\% & 6\\
        256 & 3 & 81,340 & 12 & 99.99\% & 5\\
        512 & 3 & 22,684 & 7 & 99.99\% & 5\\
        % 1,024 & 2 & 5,956 & 2 & 99.99\% & 4\\
    \bottomrule
    \end{tabular}
\end{table}

\section{Why Are PGM-Indexes So Effective?}\label{sec:theory}
In this section, we first motivate the necessity of an index height lower than $O(\log N)$ through benchmark results ($\vartriangleright$ Section~\ref{subsec:motivation_exp}).
Then, we establish a tighter sub-logarithmic bound ($\vartriangleright$ Section~\ref{subsec:theory}). 
Finally, a case study on uniformly distributed keys is provided to further validate our theoretical analysis ($\vartriangleright$ Section~\ref{subsec:case_uniform}). 

\subsection{Motivation Experiments}\label{subsec:motivation_exp}
We construct the PGM-Indexes and B+-trees using various configurations, with index statistics summarized in Table~\ref{tab:index_height_experiment}. 
Intuitively, a B+-tree with a fan-out of $B=\epsilon$ can be considered analogous to a PGM-Index where $\epsilon_i=\epsilon_\ell=\epsilon$, since a B+-tree index guarantees the search key to be located within a data block of size $B$. 

As shown in Table~\ref{tab:index_height_experiment}, we first fix $B=\epsilon=16$ while varying the input data size across $\{10^3,10^4,\cdots,10^9\}$ using synthetic uniform keys. 
As the data size $N$ increases, the height of a B+-tree index ($H_B$) follows a logarithmic growth pattern, adhering to the formula $H_B=\lceil1+\log_B\frac{N+1}{2}\rceil$. 
On the other hand, the PGM-Index height ($H_{PGM}$) grows at a much slower, sub-logarithmic rate. 
Besides, on dataset \texttt{books}, when varying $\epsilon$ within $\{2^2,2^3,\cdots,2^{10}\}$, the results consistently demonstrate that $H_{PGM}\ll H_{B}$ holds across all $\epsilon$ configurations.
Moreover, the decrease in $H_{PGM}$ relative to $\epsilon$ is also notably slower than that of $H_{B}$. 
% Furthermore, when fixing the input dataset as \texttt{books} (800 M) and varying $\epsilon$ within $\{2^2,2^3,\cdots,2^{10}\}$, 
% the results in Table~\ref{tab:index_height_experiment} consistently demonstrate that $H_{PGM}\ll H_{B}$ holds across all $\epsilon$ settings, and the decrease in $H_{PGM}$ relative to $\epsilon$ is also notably slower than that of $H_{B}$. 
In addition to the index height, we also report the numbers of leaf and internal segments in Table~\ref{tab:index_height_experiment}. 
Unlike B+-trees or other BST variants, the PGM-Index exhibits a highly ``flat'' structure, with most of the line segments (up to \textbf{99.99\%}) located at the bottom level, aligning with its slow height growth. 

Notably, the results obtained from different datasets and additional $\epsilon$ configurations are similar and therefore omitted here due to page limitations. 
The complete results are available at~\cite{pgm++}. 

\subsection{Theoretical Analysis}\label{subsec:theory}
% The results in Section~\ref{subsec:motivation_exp} reveal that a \emph{logarithmic} index height, as introduced in Theorem~\ref{theorem:pgm_height}, is too loose. 
In this section, we aim to provide a new bound tightening the previous results. 
The road map for establishing our theoretical results is outlined below.

\noindent\ding{182} Lemma~\ref{lemma:coverage_rec} and Lemma~\ref{lemma:coverage_level_i} provide a lower bound for the expected segment coverage in an arbitrary level of the PGM-Index;

\noindent\ding{183} Theorem~\ref{theorem:pgm_height_new} derives the PGM-Index height as $O(\log\log N)$, indicating sub-logarithmic growth w.r.t.~the data size $N$;

\noindent\ding{184} Theorem~\ref{theorem:complexity} concludes the space and time complexities of the PGM-Index as summarized in Table~\ref{tab:theory}. 

Notably, unless explicitly stated otherwise, the subsequent analyses adhere to the core assumptions regarding \emph{gaps} from Theorem~\ref{theorem:segment_coverage}~\cite{DBLP:conf/icml/FerraginaLV20}, that is, the gaps are i.i.d.~random variables following some unknown distribution with expectation $\mu$ and variance $\sigma^2$. 
Besides, as discussed in~\cite{DBLP:conf/icml/FerraginaLV20}, the ``i.i.d.'' assumption can be further relaxed to \emph{weakly correlated} random variables without affecting the correctness of theoretical results.

\begin{lemma}[Expected Coverage Recursion]\label{lemma:coverage_rec}
    Given a set of $N$ sorted keys $\mathcal{K}=\{k_1,\cdots,k_N\}$ and an error parameter $\epsilon$, let a random variable $C_i$ denote the number of keys in the $(i-1)$-th level that a segment in the $i$-th level can cover (i.e., satisfying the error constraint $\epsilon$). 
    Specifically, $C_0$ denotes the leaf segment coverage (i.e., level-0) on the input key set $\mathcal{K}$. 
    Then, the following recursion holds for $\mathbf{E}[C_i]$, 
    \begin{equation}
    \small
        \mathbf{E}[C_i]=\frac{\mu^2\cdot\epsilon^2}{\sigma^2}\cdot\mathbf{E}[C_0\cdot C_1\cdots C_{i-1}].
    \end{equation}
\end{lemma}
\begin{proof}
    According to the law of total expectation~\cite{chung2000course}, 
    \begin{equation}\label{eq:int_total_expectation}
    \small
        \begin{aligned}
            \mathbf{E}\left[C_i\right]=\idotsint\limits&\mathbf{E}\left[C_i\mid C_0=c_0,\cdots,C_{i-1}=c_{i-1}\right]\times \\
            &\,\,\,\, f(c_0,\cdots,c_{i-1})\, dc_0\cdots dc_{i-1},\\
        \end{aligned}
    \end{equation}
    where $f(c_0,\cdots,c_{i-1})$ is the joint probability density function of random variables $C_0,\cdots,C_{i-1}$.

    Suppose that $g$'s are the gaps of the original key set $\mathcal{K}$. 
    When fixing $C_0,\cdots,C_{i-1}$ to $c_0,\cdots,c_{i-1}$, as illustrated in Figure~\ref{fig:gaps_next_level}, w.l.o.g., an arbitrary gap in the $i$-th level, denoted by $g^{(i)}$, should be the sum of $c_0\cdot c_1\cdots c_{i-1}$ consecutive gaps on the raw key set $\mathcal{K}$\footnote{Here, we assume that all line segments within the same level exhibit equal coverage. A more rigorous analysis can be established by using concentration bounds like Chebyshev’s inequality or Chernoff bound~\cite{chung2000course}, which is omitted here for brevity.}. 
    Thus, according to Theorem~\ref{theorem:segment_coverage}, on a key set with gaps as $g^{(i)}$, the expected segment coverage (conditioned on $C_0,\cdots,C_{i-1}$) in the $i$-th level for an $\epsilon$-PLA should be, 
    \begin{equation}\label{eq:condition_expectation}
    \small
        \begin{aligned}
            \mathbf{E}[C_i|C_0=c_0,&\cdots,C_{i-1}=c_{i-1}]=\frac{\mathbf{E}\left[g^{(i)}\right]^2\cdot\epsilon^2}{\mathbf{Var}\left[g^{(i)}\right]}\\
            &=\frac{\mathbf{E}\left[\sum_{j'=j}^{j+c_0\cdot c_1\cdots c_{i-1}}g_{j'}\right]^2\cdot\epsilon^2}{\mathbf{Var}\left[\sum_{j'=j}^{j+c_0\cdot c_1\cdots c_{i-1}}g_{j'}\right]}\\
            &=(c_0\cdot c_1\cdots c_{i-1})\cdot\frac{\mu^2\cdot\epsilon^2}{\sigma^2},
        \end{aligned}
    \end{equation}
    where $\mu$ and $\sigma^2$ are the mean and variance of the gaps on the \emph{original key set} $\mathcal{K}$. 
    Taking Eq.~\eqref{eq:condition_expectation} into the integral in Eq.~\eqref{eq:int_total_expectation}, we have,
    \begin{equation}
    \small
    \begin{aligned}
        \mathbf{E}[C_i]&=\frac{\mu^2\cdot\epsilon^2}{\sigma^2}\idotsint\prod_{j=0}^{i-1}c_j\cdot f(c_0,\cdots,c_{i-1})\, dc_0\cdots dc_{i-1}\\
        &=\frac{\mu^2\cdot\epsilon^2}{\sigma^2}\cdot\mathbf{E}[C_0\cdot C_1\cdots C_{i-1}].
    \end{aligned}
    \end{equation}
    Thus, we have the statement in Lemma~\ref{lemma:coverage_rec}. 
\end{proof}

\begin{lemma}[Expected Coverage of Level-$i$]\label{lemma:coverage_level_i}
    The following lower bound holds for $\mathbf{E}[C_i]$, 
    \begin{equation}\label{eq:ci_lower_bound}
    \small
        \mathbf{E}\left[C_i\right]\geq\left(\frac{\mu^2\cdot\epsilon^2}{\sigma^2}\right)^{2^i}.
    \end{equation}
\end{lemma}
\begin{proof}
    We prove Lemma~\ref{lemma:coverage_level_i} using mathematical induction. 
    \ding{182}~\textbf{Base Case ($i'=0$):}  
    According to Theorem~\ref{theorem:segment_coverage}, $\mathbf{E}[C_0]={\mu^2\epsilon^2}/{\sigma^2}$, satisfying the inequality in Eq.~\eqref{eq:ci_lower_bound}. 
    
    \noindent\ding{183}~\textbf{Inductive Step:} 
    Assume that the lower bound in Eq.~\eqref{eq:ci_lower_bound} holds for $i'=i-1$, i.e., 
    \begin{equation}\label{eq:case_i-1}
    \small
        \mathbf{E}\left[C_{i-1}\right]\geq\left(\frac{\mu^2\cdot\epsilon^2}{\sigma^2}\right)^{2^{i-1}}.
    \end{equation}
    Then, for the case of $i'=i$, according to Lemma~\ref{lemma:coverage_rec}, we have,
    \begin{equation}\label{eq:case_i}
    \small
        \begin{aligned}
        \mathbf{E}[C_i]&=\frac{\mu^2\cdot\epsilon^2}{\sigma^2}\cdot\mathbf{E}[C_0\cdot C_1\cdots C_{i-1}]\\
        &\geq \frac{\mu^2\cdot\epsilon^2}{\sigma^2}\cdot\mathbf{E}[C_0\cdot C_1\cdots C_{i-2}]\cdot\mathbf{E}[C_{i-1}],
        % &\geq\underbrace{\frac{\mu^2\cdot\epsilon^2}{\sigma^2}\cdot\mathbf{E}[C_0\cdot C_1\cdots C_{i-2}]}_{\mathbf{E}[C_{i-1}]}\cdot\mathbf{E}[C_{i-1}],
    \end{aligned}
    \end{equation}
    considering that $C_{i-1}$ is positively correlated with $C_0\cdot C_1\cdots C_{i-2}$. 
    By the inductive hypothesis (i.e., Eq.~\eqref{eq:case_i-1}), we have,
    \begin{equation}
    \small
            \mathbf{E}[C_i]\geq\mathbf{E}[C_{i-1}]^2\geq\left(\frac{\mu^2\cdot\epsilon^2}{\sigma^2}\right)^{2^i},
    \end{equation}
    which satisfies the lower bound for $i'=i$. 
    Thus, by induction, we conclude that Lemma~\ref{lemma:coverage_level_i} holds for all $i$. 
\end{proof}

\begin{figure}[t]
    \centering
    \includegraphics[width=0.4\textwidth]{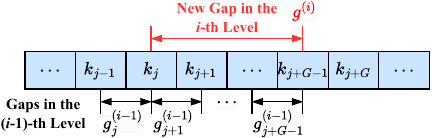}
    \caption{Illustration of gaps for the next level. Suppose $G$ is the segment coverage for the current level. The new gap in the $i$-th level is $g^{(i)}=\sum_{j'=j+1}^{j+G-1}g^{(i-1)}_{j'}$ where $g^{(i-1)}_{j'}$ is the $j'$-th gap in the $(i-1)$-th level.}  
    \label{fig:gaps_next_level}
\end{figure}

\begin{theorem}[PGM-Index Height]\label{theorem:pgm_height_new}
    Given a set $\mathcal{K}$ of $N$ sorted keys, denote the constant $G={\mu^2\epsilon^2}/{\sigma^2}$, w.h.p., the height of a PGM-Index with error parameter $\epsilon_i=\epsilon_\ell=\epsilon$ is bounded by 
    \begin{equation}\label{eq:pgm_height}
    \small
        H_{PGM}=O(\log_2\log_G N)=O(\log\log N).
    \end{equation}
\end{theorem}
\begin{proof}
    Here we only provide an intuitive proof sketch due to the page limit. 
    A more rigorous proof can be established by employing a similar technique as introduced in Theorem~4 of~\cite{DBLP:conf/icml/FerraginaLV20}. 
    
    According to Definition~\ref{def:pgm_construction}, the construction of a PGM-Index terminates when the current level consists of exactly one line segment (i.e., reaching the root level). 
    Intuitively, the index height $H_{PGM}$ can be solved by letting
    \begin{equation}\label{eq:height_equation_to_solve}
    \small
        \frac{N}{\prod_{i=0}^{H_{PGM}-1}\mathbf{E}[C_{i}]}=O(1).
    \end{equation}
    According to Theorem~\ref{lemma:coverage_level_i}, we have, 
    \begin{equation}
    \small
    \begin{aligned}
        \prod_{i=0}^{H_{PGM}-1}\mathbf{E}[C_{i}]\geq\prod_{i=0}^{H_{PGM}-1}G^{2^i}&\geq G^{\sum_{i=0}^{H_{PGM}-1}2^i}\geq G^{2^{H_{PGM}}}.
    \end{aligned}
    \end{equation}
    Thus, Eq.~\eqref{eq:height_equation_to_solve} can be solved by $H_{PGM}=O(\log_2\log_G N)$. 
\end{proof}

\begin{theorem}[Space and Time Complexity]\label{theorem:complexity}
    Given a set $\mathcal{K}$ of $N$ sorted keys, a PGM-Index with $\epsilon_i=\epsilon_\ell=\epsilon$ can process an index lookup query in $O(\log\log N)$ time using $O(N/G)$ space. 
\end{theorem}
\begin{proof}
    According to Definition~\ref{def:pgm_construction} and Example~\ref{eg:pgm_query}, querying a PGM-Index requires $H_{PGM}$ times search operations, each within a range of $2\cdot\epsilon+1$. 
    According to Theorem~\ref{theorem:pgm_height_new}, the total index lookup time should be $O(H_{PGM}\cdot\log_2(2\cdot\epsilon+1))=O(\log\log N)$.

    We further analyze the space complexity of a PGM-Index, specifically the total number of line segments required to satisfy the error constraint $\epsilon$. 
    According to Definition~\ref{def:pgm_construction} and Lemma~\ref{lemma:coverage_level_i}, the $h$-th level contains at most $N/\prod_{i=0}^{h}G^{2^i}$ line segments. 
    Thus, the upper bound on the total number of segments can be derived as,
    \begin{equation}
    \small
        \begin{aligned}
            \sum_{h=0}^{H_{PGM}-1}\frac{N}{\prod_{i=0}^{h}G^{2^i}}&\leq \sum_{h=0}^{H_{PGM}-1}\frac{N}{G^{h+1}}\leq N\cdot\frac{1-\frac{1}{G^{H_{PGM}}}}{G-1}\\
                &\leq\frac{N}{G-1}=O\left({N}/{G}\right),
        \end{aligned}
    \end{equation}
    considering that $\prod_{i=0}^{h}G^{2^i}\geq \prod_{i=0}^{h}G^{2^0}\geq G^{h+1}$. 
    % \begin{equation}
    %     \begin{aligned}
    %         \sum_{h=0}^{H_{PGM}-1}\frac{N}{\prod_{i=0}^{h}C^{2^i}}&\leq \sum_{h=0}^{H_{PGM}-1}\frac{N}{C^{h+1}}\\
    %         &\leq N\cdot \frac{1-\frac{1}{C^{H_{PGM}}}}{C-1}\\
    %         &\leq \frac{N}{C-1}=O\left(\frac{N}{C}\right).
    %     \end{aligned}
    % \end{equation}
    % Thus, we finally obtain the results in Theorem~\ref{theorem:complexity}.
\end{proof}

\begin{table}[t]
\caption{PGM-Index statistics on 10 million synthetic uniform keys with different ranges.}
\label{tab:pgm_uniform_different_range}
\small
\begin{tabular}{c|ccccc}
\toprule
\cellcolor[HTML]{f4f5f6}\textbf{Key Range}          & \cellcolor[HTML]{f4f5f6}$\epsilon$ & \cellcolor[HTML]{f4f5f6}\textbf{Height} & \cellcolor[HTML]{f4f5f6}\textbf{Segments} & \cellcolor[HTML]{f4f5f6}\textbf{Memory} & \cellcolor[HTML]{f4f5f6}$\overline{Cov}$ \\\midrule
\multirow{4}{*}{\begin{tabular}[c]{@{}c@{}}$[0,10^8]$\\ $\mu=10$\\ $\sigma^2=100.19$\end{tabular}}  & 4          & 4       & 129,503 & 2,078 KiB & 77 \\
                             & 8          & 3       & 37,732  & 604 KiB   & 265 \\
                             & 16         & 3       & 10,224  & 163 KiB   & 978 \\
                             & 32         & 2       & 2,666   & 42 KiB    & 3,751 \\\midrule
\multirow{4}{*}{\begin{tabular}[c]{@{}c@{}}$[0,10^9]$\\ $\mu=100$\\ $\sigma^2=10007.7$\end{tabular}}  & 4          & 3       & 129,659 & 2,080 KiB & 77 \\
                             & 8          & 3       & 37,601  & 602 KiB   & 266   \\
                             & 16         & 3       & 10,124  & 162 KiB   & 988   \\
                             & 32         & 2       & 2,665   & 42 KiB    & 3,752 \\\midrule
\multirow{4}{*}{\begin{tabular}[c]{@{}c@{}}$[0,10^{10}]$\\ $\mu=1000$\\ $\sigma^2=999750$\end{tabular}} & 4        & 3       & 129,586 & 2,079 KiB & 77  \\
                             & 8          & 3       & 37,597  & 602 KiB   & 266  \\
                             & 16         & 3       & 10,217  & 164 KiB   & 979 \\
                             & 32         & 2       & 2,646   & 42 KiB   & 3,779 \\
\bottomrule
\end{tabular}
\end{table}

% \begin{corollary}[Tightness of Space Bound]
%     The space complexity of a PGM-Index $O(N/G)$ is tight. 
% \end{corollary}
% \begin{proof}
    
% \end{proof}

\subsection{Case Study: Uniform Keys}\label{subsec:case_uniform}
Previously, we assume that \emph{gaps} are drawn from an \emph{unkonwn} distribution. 
To further validate the correctness of our theoretical results, we now provide a case study on uniformly distributed \emph{keys}. 

Given a key set $\mathcal{K}$, assume that all keys $k\in\mathcal{K}$ are i.i.d.~samples drawn from a uniform distribution $\mathbf{U}(\alpha, \beta)$. 
In this case, the $i$-th gap on $\mathcal{K}$ can be defined as $g_i=k_{(i)}-k_{(i-1)}$ where $k_{(i)}$ and $k_{(i-1)}$ are the $i$-th and $(i-1)$-th \emph{order statistics} of $\mathcal{K}$ (i.e., the $i$-th and $(i-1)$-th smallest values in $\mathcal{K}$). 
Then, for an arbitrary $i=2,\cdots,N$, it can be shown that $g_i$ follows a beta distribution, $g_i\sim(\beta-\alpha)\cdot\mathbf{Beta}(1, N)$, with the following mean and variance,
\begin{equation}\label{eq:uniform_keys_stats}
\small
    \begin{aligned}
        \mathbf{E}[g_i] = \frac{\beta-\alpha}{N+1}, \,\,
        \mathbf{Var}[g_i] = \frac{(\beta-\alpha)^2\cdot N}{(N+1)^2\cdot(N+2)}\approx\frac{(\beta-\alpha)^2}{(N+1)^2}.
    \end{aligned}
\end{equation}
According to Eq.~\eqref{eq:uniform_keys_stats}, the constant $G=\frac{\mu^2\cdot\epsilon^2}{\sigma^2}=\epsilon^2$, which is interestingly independent of the original key distribution. 
By Theorem~\ref{theorem:pgm_height} and Theorem~\ref{theorem:complexity}, this result implies that, for uniformly distributed keys, a PGM-Index should have the \emph{same} index height and memory footprint as long as $N$ and $\epsilon$ remain unchanged.

Table~\ref{tab:pgm_uniform_different_range} reports the statistics for PGM-Indexes constructed on three synthetic uniform key sets with different ranges. 
The empirical results further validate the correctness of the aforementioned analysis, given that the index height and segment count remain consistent across different data ranges, depending solely on the value of error constraint $\epsilon$. 

\noindent\textbf{Extension to Arbitrary Key Distributions.} 
Suppose that the keys $k_1,\cdots,k_N$ are $N$ i.i.d.~random samples drawn from an \emph{arbitrary} distribution with cumulative distribution function $F(x)$ and density function $f(x)$. 
As the $i$-th gap $g_i=k_{(i)}-k_{(i-1)}$, the distribution characteristics like mean and variance of $g_i$ can be derived by evaluating the joint density function $f_{k_{(i-1)},k_{(i)}}(x,y)$ of two consecutive order statistics $k_{(i-1)}$ and $k_{(i)}$~\cite{chung2000course}. 
For example, the expectation $\mathbf{E}[g_i]$ can be derived as,
\begin{equation}\label{eq:general_gap_distribution}
\small
    \begin{aligned}
        \mathbf{E}[g_i]&=\iint (y-x)\cdot f_{k_{(i-1)},k_{(i)}}(x,y) dxdy,\\
        f_{k_{(i-1)},k_{(i)}}(x,y)&=\frac{N!\cdot F(x)^{i-2}(1-F(y))^{N-i}f(x)f(y)}{(i-2)!(N-i)!}.
    \end{aligned}
\end{equation} 
Notably, in most cases, no closed-form solution exists for Eq.~\eqref{eq:general_gap_distribution}. 
Thus, an empirical CDF (ECDF) based on random sampling can be applied to obtain a provably accurate approximation according to the DKW~bound~\cite{dvoretzky1956asymptotic}. 

%% file: sections/4-ineffective.tex
\section{Why Are PGM-Indexes Ineffective?}\label{sec:ineffective}
The theoretical findings in Section~\ref{sec:theory} reveal that PGM-Indexes can achieve the best space-time trade-off among existing learned indexes. 
However, according to recent benchmarks~\cite{DBLP:journals/pvldb/MarcusKRSMK0K20,DBLP:journals/pvldb/WongkhamLLZLW22}, an optimized RMI~\cite{DBLP:conf/sigmod/KraskaBCDP18} consistently outperforms the PGM-Index by 20\%--40\%. 
Motivated by this, in this section, we aim to answer another critical question: Why do PGM-Indexes underperform in practice? 

\noindent\textbf{A Simple Cost Model.} We begin by introducing a simplified cost model for an arbitrary $(\epsilon_i,\epsilon_\ell)$-PGM-Index. 
Recalling the PGM-Index structure as shown in Figure~\ref{fig:pgm}, the total index lookup time for a search key $k$ can be modeled as the summation of the internal search cost with error constraint $\epsilon_i$ and the last-mile search cost with error constraint $\epsilon_\ell$, i.e.,
\begin{equation}\label{eq:pgm_cost_model}
    \begin{aligned}
        Cost &= Cost_{\text{internal}} + Cost_{\text{last-mile}} \\
        &= (H_{PGM}-1)\cdot C_{S}(\epsilon_i) + C_{S}(\epsilon_\ell) + H_{PGM}\cdot C_{L},
    \end{aligned}
\end{equation}
where $H_{PGM}$ is the index height, $C_S(\epsilon)$ represents the search cost within the range of $2\cdot\epsilon+1$, and $C_L$ is the overhead to evaluate a linear function $y=a\cdot x+b$. 

\begin{figure}[t]
    \centering
    \includegraphics[width=0.36\textwidth]{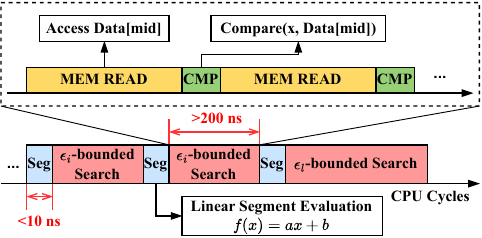}
    \caption{Illustration of the CPU cycles used for searching an $(\epsilon_i,\epsilon_\ell)$-PGM-Index with a standard binary search algorithm for the internal error-bounded search operation.}
    \label{fig:pgm_cpu}
\end{figure}

\noindent\textbf{Bottleneck: Error-bounded Search.} According to Theorem~\ref{theorem:pgm_height}, the index height $H_{PGM}=O(\log\log N)$, implying that very few internal searches are required (generally fewer than \textbf{5} for 1 billion keys). 
Additionally, as depicted in Figure~\ref{fig:pgm_cpu}, our benchmark results across various datasets and platforms indicate that evaluating a linear function typically takes \textbf{less than 10 ns}. 
In contrast, performing a search with $\epsilon=64$ takes time \textbf{more than 200 ns} by adopting a standard binary search implementation (e.g., \texttt{std::lower\_bound}). 
Based on this observation, the cost model in Eq.~\eqref{eq:pgm_cost_model} can be simplified by neglecting the segment evaluation overhead, i.e.,
\begin{equation}\label{eq:pgm_cost_model_revised}
    Cost \approx (H_{PGM}-1)\cdot C_{S}(\epsilon_i) + C_{S}(\epsilon_\ell).
\end{equation}

The revised cost model reveals that searching a PGM-Index is dominated by performing $H_{PGM}$ times error-bounded searches, which are generally known as \emph{memory-bound} operations~\cite{DBLP:journals/cacm/WilliamsWP09}. 
As depicted in Figure~\ref{fig:pgm_cpu}, an $\epsilon$-bounded binary search typically involves $\lceil\log_2 (2\cdot\epsilon+1)\rceil$ comparisons and memory accesses. 
Each comparison generally requires a few nanoseconds, whereas each memory access, if cache missed, can take approximately 100 nanoseconds due to the asymmetric nature of the memory hierarchy. 

\begin{figure}[t]
    \centering
    \includegraphics[width=0.34\textwidth]{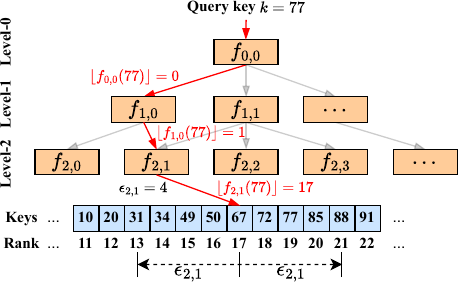}
\caption{Illustration of a 3-layer RMI~\cite{DBLP:conf/sigmod/KraskaBCDP18}. $f_{i,j}$ denotes the $j$-th model in the $i$-th layer. The path in red denotes the index traversal from the root model $f_{0,0}$. Notably, the root level is specified as level-0, opposite to the PGM-Index. }
\label{fig:rmi}
\end{figure}

\noindent\textbf{Comparison to RMI.} 
We then investigate why RMI practically outperforms the PGM-Index. 
As illustrated in Figure~\ref{fig:rmi}, the major structural difference between RMI and PGM-Index lies in their internal search mechanisms. 
For RMI, the model prediction $f_{i,j}(k)$ directly serves as the \emph{model index} for the next level (i.e., the $(i+1)$-th level), thereby bypassing the costly internal error-bounded search used in PGM-Index. 
To ensure lookup correctness, the models in the bottom layer materialize the \emph{maximum search error} to perform a last-mile error-bounded search, similar to the PGM-Index.

Table~\ref{tab:rmi_vs_pgm_details} presents the detailed overheads when querying RMI and PGM-Index. 
Consistent with previous benchmark results~\cite{DBLP:journals/pvldb/MarcusKRSMK0K20}, an optimized RMI implementation~\cite{rmi} outperforms the PGM-Index in terms of total index lookup time across all datasets. 
Specifically, for PGM-Index, the internal search time $T_i$ accounts for \textbf{69\%--81\%} of the total index lookup overhead; 
in contrast, for RMI, this ratio is as low as \textbf{19\%--27\%}, supporting our earlier claim. 

% pgm 69% 69% 81% 79% 
% rmi 23% 23% 27% 19%

\noindent\textbf{Is RMI the Best Choice?}
While RMI generally outperforms the PGM-Index, its design poses a critical limitation: RMI is hard to guarantee a maximum error before index construction, making its performance highly \emph{data-sensitive}. 
As shown in Table~\ref{tab:rmi_vs_pgm_details}, RMI's maximum error ranges from \textbf{63} to $\mathbf{3.1\times 10^5}$, resulting in high \emph{worst-case} last-mile search overhead. 
Such ``unpredictability'' also raises the challenge of building an accurate cost model for RMI-like indexes, which is crucial for practical DBMS to perform effective cost-based query optimization~\cite{DBLP:journals/csur/JarkeK84}. 
Moreover, given the identified bottleneck in querying a PGM-Index, a natural idea is to accelerate the costly internal error-bounded search operation. 
In Section~\ref{sec:optimization}, we demonstrate how a simple hybrid branchless search strategy can make the ``ineffective'' PGM-Index outperform RMI. 

%% file: sections/5-improvement.tex
\section{PGM++: Optimization to PGM-Index}\label{sec:optimization}
This section introduces PGM++, a \emph{simple yet effective} variant of the PGM-Index by incorporating a hybrid error-bounded search strategy ($\vartriangleright$~Section~\ref{subsec:search}) and an automatic parameter tuner based on well-calibrated cost models ($\vartriangleright$~Section~\ref{subsec:cost_model}). 
% We first propose a hybrid search strategy to replace the standard binary search used in PGM-Indexes (Section~\ref{subsec:search}). 
% Then, we establish a cost model to estimate the space and time overhead of PGM-Index (Section~\ref{subsec:cost_model}). 
% Guided by the cost model, we further design a set of parameter tuning strategies to properly set $\epsilon_i$ and $\epsilon_\ell$ (Section~\ref{subsec:parameter_tuning}). 

\subsection{Hybrid Search Strategy}\label{subsec:search}
As the error-bounded search operation is identified as the bottleneck in querying the PGM-Index, our optimized structure, named PGM++, employs a \emph{hybrid search strategy} to replace the standard binary search. 
To start, we discuss the impact of branch misses in standard binary search implementation. 

\noindent\textbf{Branch Prediction and Branch Miss.} 
Modern CPUs rely on sophisticated branch predictors to enhance pipeline parallelism by forecasting the outcomes of conditional jump instructions (e.g., \texttt{JLE} and \texttt{JAE} instructions in the \texttt{X86} architecture). 
These predictors are highly effective for simple, repetitive tasks such as for loops or pointer chasing, where the pattern of execution is predictable~\cite{hennessy2011computer}. 
However, in the case of standard binary search implementations, such as the widely used \texttt{std::lower\_bound}, branching exhibits a \emph{random pattern}, leading to a high branch miss rate of approximately \textbf{50\%}~\cite{DBLP:journals/pvldb/SchulzBS18}.
As depicted in Figure~\ref{fig:cpu_pipeline}(a), a branch miss stalls the entire CPU pipeline until the branch condition is resolved (e.g., the comparison \texttt{d[mid]>=k} in line~5 of function \texttt{lower\_bound}).

% To fully benefit from the hardware pipeline parallelism, modern CPUs are equipped with sophisticated branch predictors to guess whether a conditional jump (e.g., \texttt{JLE}/\texttt{JAE} instructions on \texttt{x86}) will be executed or not. 
% Pipeline efficiency can be guaranteed if the branch predictor makes accurate predictions, which typically holds on simple workloads like \texttt{for} loops or pointer chasing~\cite{hennessy2011computer}. 
% However, it is generally known that the standard binary search implementation like \texttt{std::lower\_bound} in C++ exhibits a \emph{fully random} branch pattern, yielding a high branch missing rate ($\approx \mathbf{50\%}$)~\cite{DBLP:journals/pvldb/SchulzBS18}. 

\noindent\textbf{Branchless Binary Search.} 
A simple optimization~\cite{DBLP:journals/pvldb/SchulzBS18} to the standard binary search is to \emph{remove} the branches by conditional move instructions (e.g., \texttt{CMOV} on \texttt{X86} and \texttt{MOVGE} on \texttt{ARM}), which allow \emph{both} sides of a branch to execute and keeps the valid one based on the evaluated condition.
% \footnote{On \texttt{ARM}, most instructions are \emph{conditional} like \texttt{MOVGE}~\cite{arm-asm}, serving a similar purpose.}
As illustrated in   Figure~\ref{fig:cpu_pipeline}(b), eliminating branches (function \texttt{lower\_bound\_brl}) maximizes the CPU pipeline utilization, yielding up to a \textbf{51\%} reduction in total search time. 
Notably, \texttt{CMOV} is not the ``silver bullet'' as it disables the native branch predictor and incurs extra overhead due to its intrinsic complexity. 
On large datasets (>LLC size), the performance gap between branchy and branchless searches diminishes as the memory access latency dominates the total overhead. 
However, such extra overhead is \emph{negligible} particularly when the search range fits within the L2 cache, making \texttt{CMOV} performance-worthy in PGM-Index (usually $\epsilon\leq 1024$). 

% Thus, in some cases, eliminating all branches may be even more costly. 

\begin{table}[t]
\caption{Query processing details. For PGM-Index, $\epsilon_i=16$ and $\epsilon_\ell=32$. For RMI, we adopt CDFShop~\cite{DBLP:conf/sigmod/MarcusZK20} to find an optimized RMI structure with a comparable space to the PGM-Index. }
\label{tab:rmi_vs_pgm_details}
\small
\begin{tabular}{c|cccccc}
\toprule
\cellcolor[HTML]{f4f5f6}\textbf{Data}       & \cellcolor[HTML]{f4f5f6}\textbf{Index} & \multicolumn{1}{c}{\cellcolor[HTML]{f4f5f6}\textbf{Size}} & \cellcolor[HTML]{f4f5f6}\textbf{\begin{tabular}[c]{@{}c@{}}Max\\ Err.\end{tabular}} & \cellcolor[HTML]{f4f5f6}\textbf{\begin{tabular}[c]{@{}c@{}}Internal\\ Time\end{tabular}} & \cellcolor[HTML]{f4f5f6}\textbf{\begin{tabular}[c]{@{}c@{}}Last-mile\\ Time\end{tabular}} & \cellcolor[HTML]{f4f5f6}\textbf{Total} \\\midrule
\multirow{2}{*}{\texttt{fb}} & \texttt{PGM} & 16.1 MB & \textbf{32} & 675 ns & \textbf{300 ns} & 975 ns \\
                       & \texttt{RMI} & 24.0 MB & 568 & \textbf{185 ns} & 614 ns & \textbf{799 ns} \\\midrule
\multirow{2}{*}{\texttt{wiki}}  & \texttt{PGM} & 1.3 MB & \textbf{32} & 606 ns & \textbf{270 ns} & 876 ns \\
                       & \texttt{RMI} & 1.0 MB & 63 & \textbf{95 ns} & 317 ns & \textbf{412 ns} \\\midrule
\multirow{2}{*}{\texttt{books}} & \texttt{PGM} & 37.6 MB & \textbf{32} & 887 ns & \textbf{208 ns} & 1095 ns \\
                       & \texttt{RMI} & 40.0 MB & 302 & \textbf{159 ns} & 429 ns & \textbf{588 ns} \\\midrule
\multirow{2}{*}{\texttt{osm}} & \texttt{PGM} & 44.4 MB & \textbf{32} & 824 ns & \textbf{212 ns} & 1036 ns \\
                       & \texttt{RMI} & 96.0 MB & 311K & \textbf{146 ns} & 636 ns & \textbf{782 ns} \\\bottomrule
\end{tabular}
\end{table}

\begin{figure}[t]
    \centering
    \vspace{-2ex}
    \includegraphics[width=0.45\textwidth]{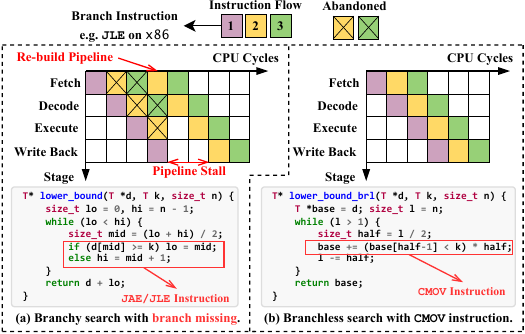}
    \caption{Illustration of the CPU pipeline status for executing (a) standard binary search (\texttt{std::lower\_bound}) and (b) branchless binary search enabled by \texttt{CMOV} instruction. }
    \label{fig:cpu_pipeline}
\end{figure}

\begin{figure}[t]
     \centering
     \begin{subfigure}[b]{0.2\textwidth}
         \centering
         \includegraphics[width=\textwidth]{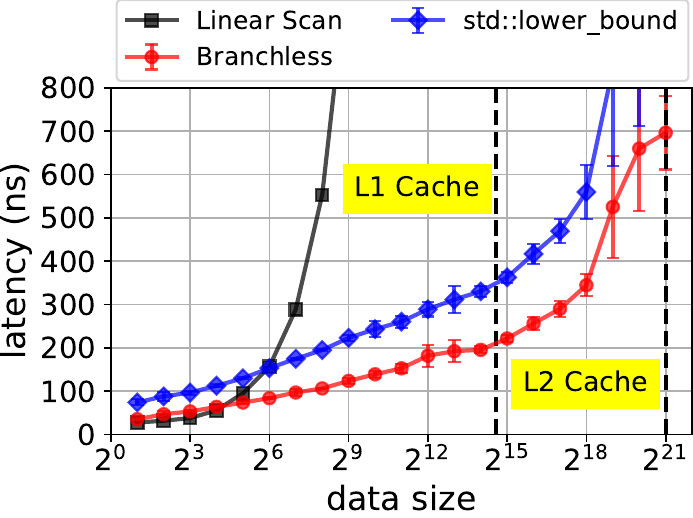}
         \caption{Platform \texttt{ARM}.}
         \label{fig:search_arm}
     \end{subfigure}
     \begin{subfigure}[b]{0.2\textwidth}
         \centering
         \includegraphics[width=\textwidth]{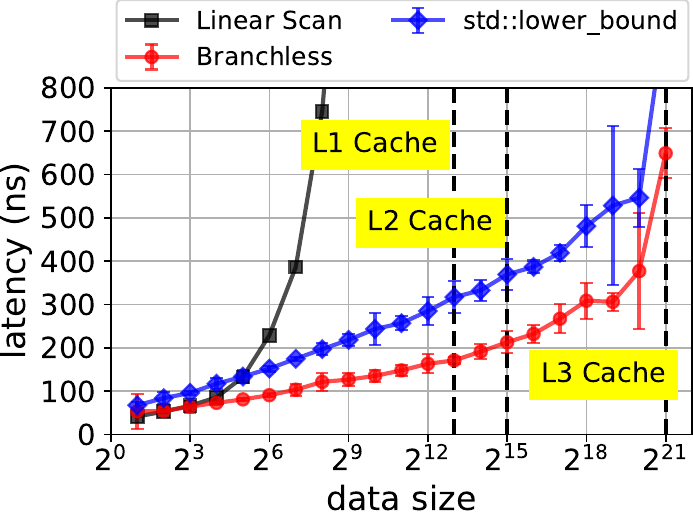}
         \caption{Platform \texttt{X86-1}.}
         \label{fig:search_x86}
     \end{subfigure}
     \vspace{1pt}
        \caption{Latency w.r.t.~data size for linear search, binary search (\texttt{std::lower\_bound}), and branchless binary search.}
        \label{fig:bench_search_results}
\end{figure}

\noindent\textbf{Benchmark Results.} 
Figure~\ref{fig:bench_search_results} presents the benchmark result for branchy binary search (\texttt{std::lower\_bound} from \texttt{STL}), branchless binary search (similar to \texttt{lower\_bound\_brl} in Figure~\ref{fig:cpu_pipeline}(b)), and linear scan, tested on synthetic \texttt{uint64\_t} key sets of varying sizes. 
The results indicate that, on both \texttt{ARM} and \texttt{X86} platforms, branchless search demonstrates superior performance across a wide range of data sizes, excluding very small sets (e.g., $N\leq 16$), where the linear scan is more efficient. 
Compared to \texttt{std::lower\_bound}, our branchless search implementation achieves a performance improvement of approximately $\mathbf{1.2\times}$ to $\mathbf{1.6\times}$.

It is noteworthy that other search algorithms, like k-ary search and interpolation search~\cite{DBLP:journals/pvldb/SchulzBS18}, are not included in this comparison. 
This is because, the search range in the PGM-Index is typically small (e.g., $\mathbf{\epsilon\leq 1024}$), where more advanced search algorithms do not \emph{consistently} outperform a branchless binary search. 
Additionally, we do not consider architecture-aware optimizations like SIMD and memory pre-fetching, as our work is intended to provide a detailed theoretical and experimental \emph{revisit} of the original PGM-Index~\cite{DBLP:journals/pvldb/FerraginaV20}. 
The simple hybrid search strategy, as described below, is sufficient to showcase the potential of PGM-Index. 

\noindent\textbf{Hybrid Search.} 
Based on the above discussion and benchmark results, our PGM++ adopts the following \emph{hybrid search} operator: 
\begin{equation*}
\small
    \mathsf{hybrid\_search} = 
    \begin{cases}
        \mathsf{linear\_scan} & \text{ if Search Range } \leq \delta\\
        \mathsf{lower\_bound\_brl} & \text{ if Search Range } > \delta
    \end{cases}
\end{equation*}
where $\delta$ is a threshold to switch to linear search (8 on \texttt{ARM}/\texttt{X86-1}, and 16 on \texttt{X86-2}). 
Then, as illustrated in Figure~\ref{fig:hybrid_search_eg}, PGM++ processes an index lookup query as follows. 
\textbf{Step} \ding{182}: Starting from the root layer, identify the layer $l$ where the \emph{next} layer's segment count exceeds $\delta$ and skip all the layers before $l$. 
\textbf{Step} \ding{183}: Starting from layer $l$, perform internal searches (using \textsf{hybrid\_search}) with error bound $\epsilon_i$ until reaching the bottom layer. 
\textbf{Step} \ding{184}: Perform last-mile search on sorted keys (using \textsf{hybrid\_search}) with error bound $\epsilon_\ell$. 
Notably, the specific search strategy for each layer can be determined at compile time, without introducing any extra runtime overhead. 

Recalling our theoretical findings in Section~\ref{sec:theory}, the height of PGM-Index grows at a sub-logarithmic rate of $O(\log\log N)$, leading to an extremely \emph{flat} hierarchical structure where the \emph{non-bottom} layers contain very few line segments. 
Due to the structural invariance of PGM-Index (as defined in Definition~\ref{def:pgm_construction}), instead of recursively searching from the root, PGM++ skips all layers until reaching the first layer whose next layer is considered \emph{dense} (segment count $>\delta$). 
This strategy, outlined in Step~\ding{182}, effectively reduces search overhead, particularly in a \emph{cold-cache} environment.

\begin{figure}[t]
    \centering
    \includegraphics[width=0.38\textwidth]{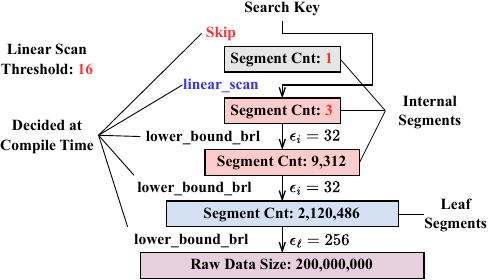}
    \caption{A toy example of the hybrid search strategy.}
    \label{fig:hybrid_search_eg}
\end{figure}

% \noindent\textbf{Performance Preview.} 
% Here we highlight the key results from the benchmark in Section~\ref{sec:exp}. 
% By adopting the hybrid search strategy, across \emph{all platforms}, our PGM++ \emph{consistently outperforms} the original PGM-Index and RMI by up to $\mathbf{1.72\times}$ and \textbf{\textcolor{red}{X}}. 
% Moreover, as PMG++ skips unnecessary searching from the root, such speedup can be more significant in the cold-cache environment where memory accessing dominates the total overhead. 

\subsection{Calibrated Cost Model}\label{subsec:cost_model}
To efficiently and effectively determine the error bounds for internal search ($\epsilon_i$) and last-mile search ($\epsilon_\ell$), we first develop cost models that estimate the space and time overheads without the need for physically constructing the PGM-Index. 

\noindent\textbf{Space Cost Model.} 
According to Section~\ref{subsec:motivation_exp} and Table~\ref{tab:index_height_experiment}, the space overhead of a PGM-Index is predominantly determined by the number of segments in the bottom layer (denoted as $L$), which accounts for up to $\mathbf{>99.9\%}$ of the total space cost. 
Therefore, to simplify the space cost model, we focus solely on the leaf segments and ignore the internal segments. 
According to the results in~\cite{DBLP:conf/icml/FerraginaLV20} (i.e., Theorem~\ref{theorem:segment_coverage}), $L\propto N\sigma^2/\epsilon_\ell^2\mu^2$, where $\mu$ and $\sigma^2$ refer to the mean and variance of \emph{gaps} on the input sorted keys, respectively. 

However, as discussed in Section~\ref{sec:benchmark_setting}, this estimation, which relies on the \emph{global} gap distribution, is often too \emph{coarse} for practical datasets due to the inherent heterogeneity in gap distributions. 
To develop a more fine-grained cost model, we partition the gaps into a set of consecutive and disjoint chunks $\mathcal{P}$ (with $\sum_{P\in\mathcal{P}}|P|=N$) by using a kernel-based change-point detection algorithm~\cite{DBLP:journals/jmlr/ArlotCH19}. 
Assuming that gaps are identically distributed within each partition $P\in\mathcal{P}$, the refined estimator for $L$ becomes: 
\begin{equation}\label{eq:space_cost}
\small
    L(\epsilon_\ell) \propto \sum\nolimits_{P\in\mathcal{P}} {N_P\sigma_{P}^2}/{\epsilon_\ell^2\mu_{P}^2},
\end{equation}
where $N_P$, $\mu_P$, and $\sigma^2_P$ represent the size, mean, and variance of gaps within partition $P\in\mathcal{P}$, respectively. 
The total space cost of an $(\epsilon_i,\epsilon_\ell)$-PGM-Index is then given by $M=L(\epsilon_\ell)\cdot\mathsf{sizeof}(seg)$, where $\mathsf{sizeof}(seg)$ is the number of bytes required to encode a line segment $seg=(s, a, b)$. 
Typically, $\mathsf{sizeof}(seg)=24$ for \texttt{uint64\_t} keys and \texttt{double} slope and intercept. 

\noindent\textbf{Time Cost Model.}
% We further devise a time cost model to estimate the lookup overhead on a PGM-Index by extending the simple cost model introduced in Eq.~\eqref{eq:pgm_cost_model_revised}. 
According to the discussions in Section~\ref{sec:ineffective}, the majority of the index lookup overhead comes from the recursively invoked error-bounded search operations. 
As we adopt a hybrid search strategy, the simplified cost model introduced in Eq.~\eqref{eq:pgm_cost_model_revised} can be further refined as follows,
\begin{subequations}
\small
\begin{align}
    Cost(\epsilon_i, \epsilon_\ell) &= Cost_{\text{internal}} + Cost_{\text{last-mile}}\label{eq:time_cost_total}\\
    Cost_{\text{last-mile}} &= \lceil\log_2(2\cdot\epsilon_\ell +1)\rceil\cdot C_{\text{miss}} \label{eq:last_mile_cost}\\
    Cost_{\text{internal}} &= (H_{PGM}-1)\cdot (C_S(\epsilon_i)+C_{\text{segment}})\label{eq:internal_cost_overall}\\
    C_S(\epsilon_i) &= \begin{cases}
        C_{\text{linear}} & \text{if } 2\cdot\epsilon_i+1\leq\delta\\
        \lceil\log_2(2\cdot\epsilon_i+1)\rceil\cdot C_{\text{hit}} & \text{if } 2\cdot\epsilon_i+1>\delta
    \end{cases} \label{eq:internal_cost}\\
    H_{PGM} &\propto \log_2\log_{\mu^2\epsilon_i^2/\sigma^2} L(\epsilon_\ell)\label{eq:pgm_height_est}
\end{align}
\end{subequations}
where (a) constants $C_{\text{miss}}$ and $C_{\text{last-mile}}$ are the memory access costs when missing or hitting L1/L2 cache; (b) constant $C_{\text{segment}}$ refers to the overhead of evaluating a linear function $y=a\cdot x + b$; (c) constant $C_{\text{linear}}$ is the cost of performing a linear search within the range of $2\cdot\epsilon_i+1$; and (d) $L(\epsilon_\ell)$ is the count of leaf segments estimated by Eq.~\eqref{eq:space_cost}. 
Notably, all constants in the cost model are estimated by probe datasets for each platform. 

In Eq.~\eqref{eq:last_mile_cost}, we assume a \emph{cold-cache} environment, as the raw key set $\mathcal{K}$ is large enough and the access to $\mathcal{K}$ is too random for hardware prefetchers to be effective. 
Conversely, in Eq.~\eqref{eq:internal_cost}, we assume a \emph{hot-cache} environment, since the non-bottom layers contain very few segments, making it highly likely for these segments to be cache-resident after processing a few queries. 
It is noteworthy that when index data can be well-cached by the CPU, the segment computation overhead becomes non-negligible. 
That is why Eq.~\eqref{eq:internal_cost_overall} includes an additional term $C_{\text{segment}}$.

\noindent\textbf{PGM-Index Parameter Tuning.} 
With the space and time cost models, the two error parameters, $\epsilon_i$ and $\epsilon_\ell$, can be automatically configured by minimizing the potential lookup cost while satisfying a pre-specified space constraint. 
\textbf{Step}~\ding{182}: Given a \emph{rough} index storage budget $B$, according to Eq.~\eqref{eq:space_cost}, $\epsilon_\ell$ can be estimated by
\begin{equation}\label{eq:eps_ell_opt}
\small
    \widetilde{\epsilon_\ell} = \sqrt{\frac{\mathsf{sizeof}(seg)}{B}\sum\nolimits_{P\in\mathcal{P}}{N_P\sigma_P^2}/{\mu_P^2}}.
\end{equation}
\textbf{Step}~\ding{183}: With a determined $\epsilon_\ell=\widetilde{\epsilon_\ell}$, $\epsilon_i$ can be derived by minimizing the index search overhead as formulated in Eq.~\eqref{eq:time_cost_total}, i.e.,
\begin{equation}\label{eq:eps_i_opt}
\small
    \widetilde{\epsilon_i} = {\arg\min}_{\epsilon_i\in\mathcal{E}} Cost(\epsilon_i, \widetilde{\epsilon_\ell}),
\end{equation}
where $\mathcal{E}$ is the set of possible values for $\epsilon_i$ ($\mathcal{E}=\{2^j|j=2,\cdots,10\}$ in our implementation). 
Intuitively, to minimize Eq.~\eqref{eq:time_cost_total}, $\epsilon_i$ should neither be too large nor too small. 
According to the cost model, a larger $\epsilon_i$ increases the overhead of $C_S(\epsilon_i)$ in Eq.~\eqref{eq:internal_cost}, while a smaller $\epsilon_i$ results in more layers to traverse (i.e., $H_{PGM}$ in Eq.~\eqref{eq:pgm_height_est}). 
Notably, although Eq.~\eqref{eq:eps_i_opt} has an analytical solution by solving $\frac{\partial Cost}{\partial \epsilon_i}=0$, in practice, we simply enumerate all possible $\epsilon_i\in\mathcal{E}$ to find the optimal value, as $\mathcal{E}$ is typically a small set ($|\mathcal{E}|<10$). 

Figure~\ref{fig:pgm_eps_tuning} reports the results of the observed index lookup costs w.r.t.~different values of $\epsilon_i$ and $\epsilon_\ell$. 
When fixing $\epsilon_\ell$, the time cost w.r.t.~$\epsilon_i$ exhibits a ``U''-shaped pattern, consistent with our earlier analysis based on the established cost model.

\noindent\textbf{Takeaways.} 
Our cost model for PGM++ can be easily extended to \emph{any} PGM-Index variants like~\cite{zhang2024making,DBLP:journals/pvldb/FerraginaV20}. 
In contrast to existing cost models for learned indexes (mostly based on RMI) like~\cite{DBLP:journals/pvldb/ZhangG22}, our cost model is \emph{workload-independent}, relying solely on gap distribution characteristics and platform-aware cost constants. 
These features enhance the robustness of parameter tuning, as the cost is optimized for \emph{all} queries rather than being tailored to a specific workload. 

\begin{figure}[t]
     \centering
     \begin{subfigure}[b]{0.21\textwidth}
         \centering
         \includegraphics[width=\textwidth]{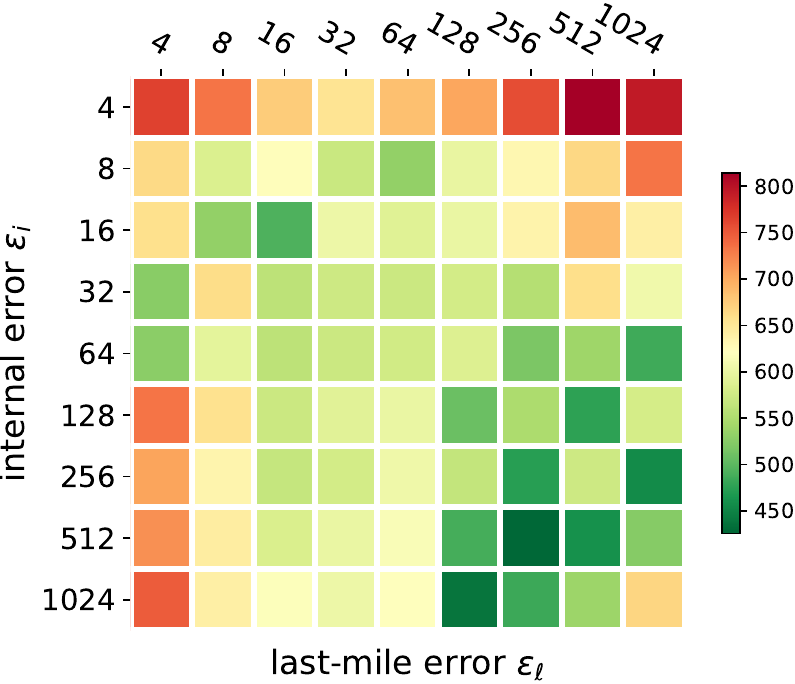}
     \end{subfigure}
     \begin{subfigure}[b]{0.21\textwidth}
         \centering
         \includegraphics[width=\textwidth]{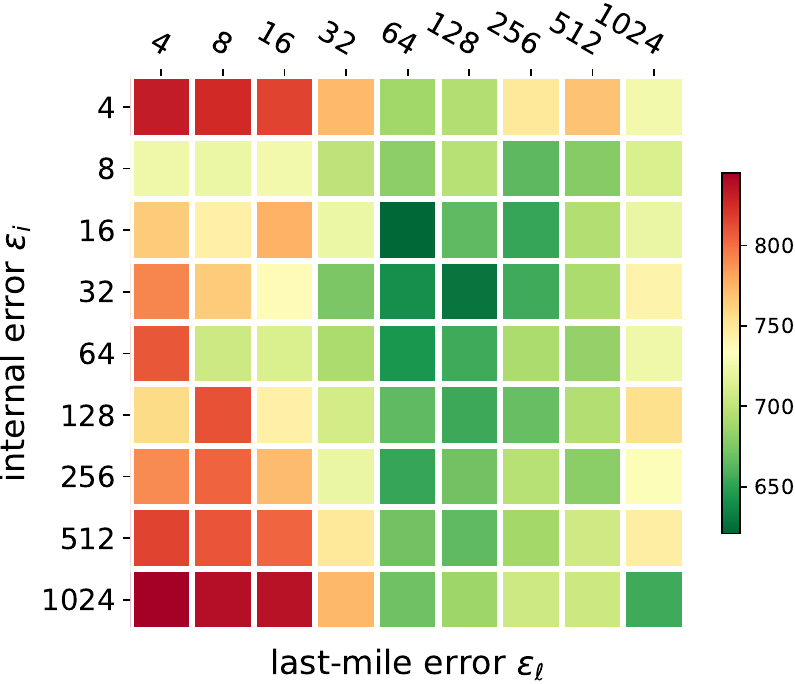}
     \end{subfigure}

     \begin{subfigure}[b]{0.21\textwidth}
         \centering
         \includegraphics[width=\textwidth]{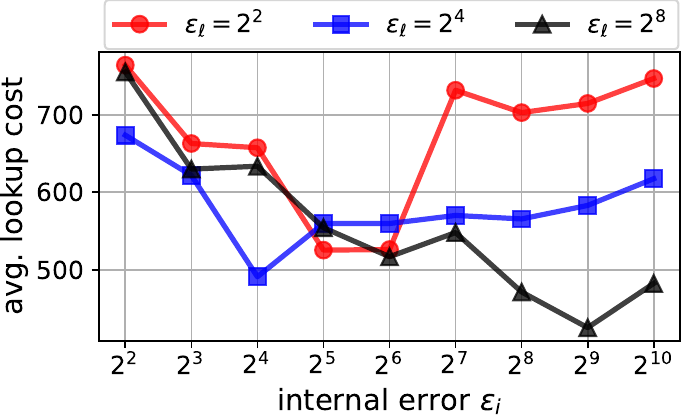}
         \caption{Lookup cost on \texttt{wiki}.}
         \label{fig:wiki_eps_tuning}
     \end{subfigure}
     \begin{subfigure}[b]{0.21\textwidth}
         \centering
         \includegraphics[width=\textwidth]{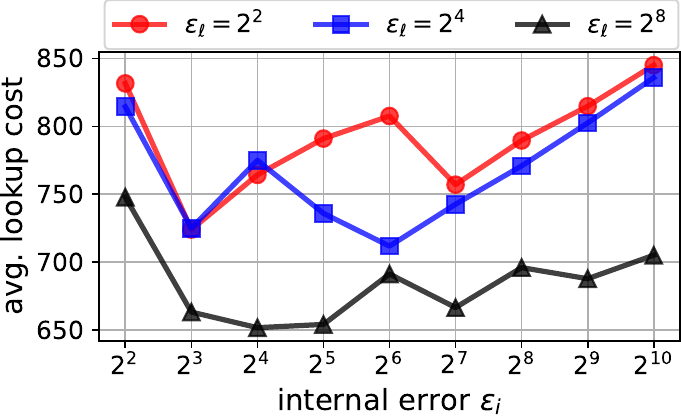}
         \caption{Lookup cost on \texttt{books}.}
         \label{fig:book_eps_tuning}
     \end{subfigure}
     \vspace{1ex}
        \caption{Observed index lookup overhead (unit: ns) of PGM++ on \texttt{x86-1} w.r.t.~different combinations of $(\epsilon_i, \epsilon_\ell)$. }
        \label{fig:pgm_eps_tuning}
\end{figure}

\begin{figure*}[t]
     \centering
     \begin{subfigure}[b]{0.98\textwidth}
         \centering
         \includegraphics[width=\textwidth]{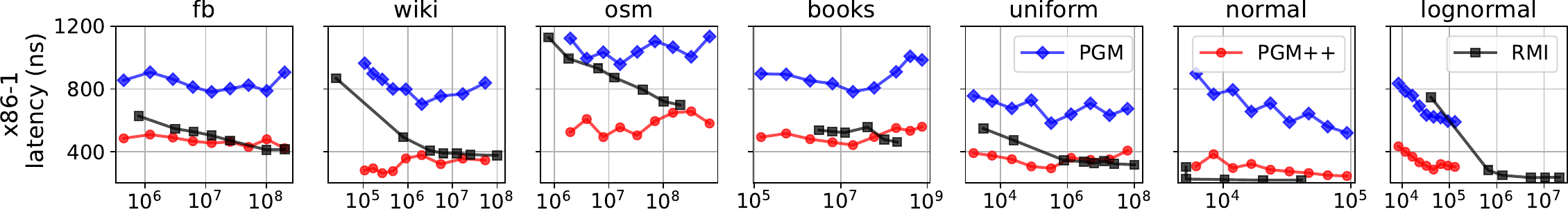}
     \end{subfigure}
     \begin{subfigure}[b]{0.98\textwidth}
         \centering
         \includegraphics[width=\textwidth]{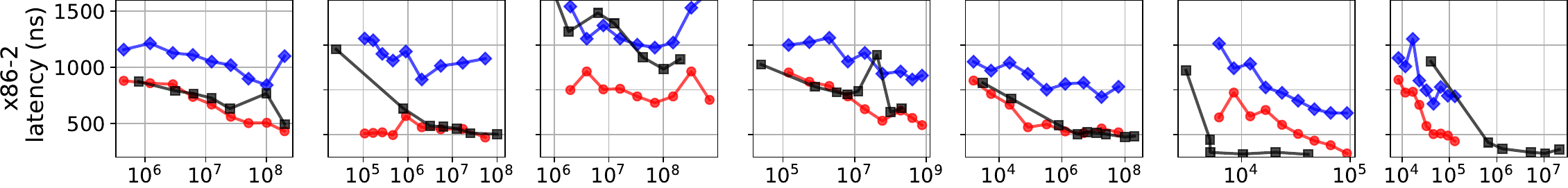}
     \end{subfigure}
     \begin{subfigure}[b]{0.98\textwidth}
         \centering
         \includegraphics[width=\textwidth]{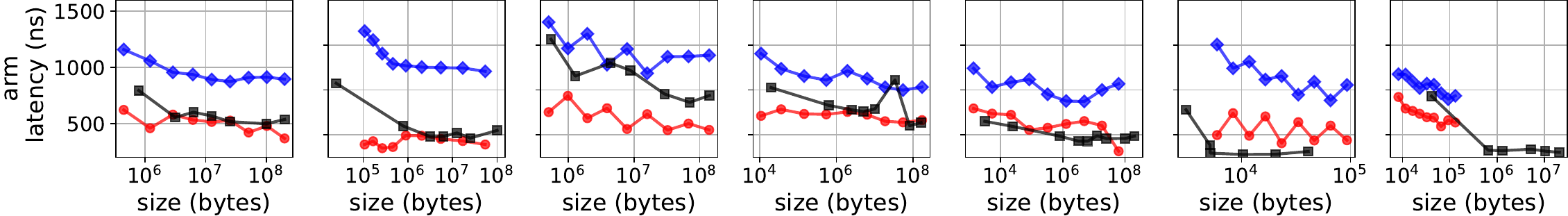}
     \end{subfigure}
     \vspace{1pt}
        \caption{Space and time tradeoffs for seven datasets on three platforms (workload: \textsf{Uniform}).}
        \label{fig:overall_evaluation}
        \vspace{-2ex}
\end{figure*}

%% file: sections/6-experiment.tex
\section{Experimental Study}\label{sec:exp}
In this section, we present the major benchmark results to answer the vital question that whether PGM++ is capable of reversing the ``ineffective'' scenario of PGM-Indexes. 
The experimental setups have been detailed in Section~\ref{sec:benchmark_setting}. 

% \noindent\ding{182} ? ($\vartriangleright$ Section~\ref{subsec:exp_overall})

% \noindent\ding{183} Whether the cost model can precisely depict the true performance of PGM-Indexes and whether the automatic parameter tuning strategy can pick error parameters to minimize the search cost? ($\vartriangleright$ Section~\ref{subsec:exp_cost_model}) 

\subsection{Overall Evaluation}\label{subsec:exp_overall}
\textbf{Baseline and Implementation.} 
We implement and evaluate three learned indexes: \ding{182} \textsf{RMI}, the optimized recursive model index~\cite{DBLP:conf/sigmod/KraskaBCDP18,DBLP:journals/pvldb/MarcusKRSMK0K20}, \ding{183} \textsf{PGM}, the original PGM-Index implementation~\cite{DBLP:journals/pvldb/FerraginaV20,pgm}, and \ding{184} \textsf{PGM++}, our optimized PGM-Index variant. 
For \textsf{RMI}, we adopt CDFShop~\cite{DBLP:conf/sigmod/MarcusZK20} to produce a set of optimal RMI configurations under various index sizes. 
For \textsf{PGM}, we construct $9\times9$ PGM-Indexes with $(\epsilon_i,\epsilon_\ell)\in\mathcal{E}\times\mathcal{E}$ and $\mathcal{E}=\{2^2,\cdots,2^{10}\}$. 
Then, for each $\epsilon_\ell\in\mathcal{E}$, the fastest PGM-Index is reported. 
Similarly, for \textsf{PGM++}, we adopt the PGM-Index configuration tuned by cost models (Section~\ref{subsec:cost_model}) for each $\epsilon_\ell\in\mathcal{E}$. 
For \textsf{PGM} and \texttt{PGM++}, according to Eq.~\eqref{eq:space_cost}, each $\epsilon_\ell$ corresponds to an index storage budget. 

We do not consider other PGM or RMI variants, such as the cache-efficient RMI~\cite{DBLP:journals/pvldb/ZhangG22} or the IO-efficient PGM-Index~\cite{zhang2024making}. 
This is because this work \emph{primarily} aims at exploring the theoretical aspect and performance bottlenecks inherent in the PGM-Index.
Our findings, however, possess a broader applicability, as they can be generalized to \emph{any} PGM-like indexes. 
This study also excludes non-learned baselines like B+-tree variants as they have been extensively compared in previous learned index benchmarks like~\cite{DBLP:journals/pvldb/MarcusKRSMK0K20,DBLP:journals/pvldb/WongkhamLLZLW22}. 

\noindent\textbf{Overall Evaluation.} 
Figure~\ref{fig:overall_evaluation} presents the trade-offs between index lookup overhead and storage cost across all seven datasets and three platforms on \textsf{Uniform} query workloads. 
The results show that, in terms of index lookup time, \textsf{PGM++} consistently outperforms \textsf{PGM} by a factor of $\mathbf{1.2\times}\sim\mathbf{2.2\times}$ with the same index size, supporting our bottleneck analysis for PGM-Indexes (Section~\ref{sec:ineffective}). 
In contrast to the optimized \textsf{RMI}, our \textsf{PGM++} addresses the costly internal index traversal through a hybrid search strategy, generally delivering better or, in some cases, comparable lookup efficiency, achieving speedups of up to $\mathbf{1.56\times}$. 
An outlier case is on dataset \texttt{normal}, \textsf{RMI} significantly outperforms \textsf{PGM++} and \textsf{PGM}. 
The reason is that the optimized \textsf{RMI}, based on CDFShop~\cite{DBLP:conf/sigmod/MarcusZK20}, adopts non-linear models (with the best \textsf{RMI} uses cubic splines), which can fit normal keys very well (maximum error <4). 
However, on other datasets, especially complex real-world datasets, \textsf{RMI} fell short in fitting the data with constrained error limits, leading to costly last-mile search overhead as discussed in Section~\ref{sec:ineffective}. 

\noindent\textbf{Space-time Trade-off.}
In most cases, \textsf{PGM++} offers the best space-time trade-off. 
However, interestingly, unlike \textsf{RMI}, whose performance improves with increased index memory usage, \textsf{PGM++} exhibits an ``irregular'' pattern in its time-space relationship. 
This is because \textsf{PGM++} is specifically optimized for query efficiency at a given storage budget. 
Leveraging accurate cost models, our parameter tuner can find configurations to provide competitive query efficiency, even under limited space constraints. 
For example, on dataset \texttt{wiki}, \textsf{PGM++} uses just \textsf{0.1 MB} of memory to outperform an \textsf{RMI} with over \textsf{100 MB} space.

\noindent\textbf{Influence of Architecture.}
From Figure~\ref{fig:overall_evaluation}, the comparison results vary across different platforms. 
For dataset \texttt{osm}, compared to \textsf{PGM}, \textsf{PGM++} achieves an average speedup ratio of $\mathbf{1.78\times}$ on \texttt{x86-1} and \texttt{arm}. 
However, such a speedup decreases to $\mathbf{1.32\times}$ on platform \texttt{x86-2}. 
This is because the memory access latency on \texttt{x86-2} is much higher than that on \texttt{x86-1}, which reduces the improvement brought by adopting the hybrid search strategy.

\noindent\textbf{Effects of Workloads.}
We also evaluate an extreme query workload, \textsf{Zipfan}, though the results are not included in this paper due to space limits. 
Queries sampled from a Zipfan distribution exhibit a highly \emph{long-tail} pattern, where the first 1K elements are frequently accessed (Section~\ref{sec:benchmark_setting}). 
Under this workload, \textsf{PGM++}, \textsf{PGM}, and \textsf{RMI} all achieve lower query latencies by up to $\mathbf{1.77\times}$, $\mathbf{2.13\times}$, and $\mathbf{4.58\times}$, respectively, compared to their performance on \textsf{Uniform} workloads. 
\textsf{RMI} shows the most substantial gains, as the last-mile search cost dominates the total index lookup time ($\mathbf{>90\%}$).
This phase benefits greatly from the spatial locality inherent in the \textsf{Zipfan} workload, where frequently accessed memory is more likely to be cached.

\subsection{Cost Model and Parameter Tuner}\label{subsec:exp_cost_model}
\noindent\textbf{Space Cost Model.} 
According to Section~\ref{subsec:cost_model}, the leaf segment count ($L$) dominates the PGM-Index space cost. 
Here, we evaluate three different segment count estimators: (a) \textsf{SIMPLE}, which directly applies Theorem~\ref{theorem:segment_coverage} on the \emph{entire} gap distribution; 
(b) \textsf{CLIP}, which applies Theorem~\ref{theorem:segment_coverage} on the gaps excluding extreme values (<0.01-quantile or >0.99-quantile); 
and (c) \textsf{ADAP}, which partitions gaps into disjoint chunks and aggregates the segment count estimated for each chunk (as in Eq.~\eqref{eq:space_cost}). 

As shown in Figure~\ref{fig:space_est}, compared to the true segment count (\textsf{TRUE}), \textsf{ADAP} consistently achieves accurate estimations across all seven datasets, nearly overlapping the \textsf{TRUE} line. 
In addition, excluding uniformly distributed datasets (e.g., \texttt{books} and \texttt{uniform}), \textsf{SIMPLE} performs the worst, validating our discussion in Section~\ref{sec:benchmark_setting} that extreme gap values significantly affect estimation accuracy. 
Notably, \textsf{CLIP} also delivers accurate results on real datasets \texttt{fb}, \texttt{wiki}, and \texttt{osm}. 
This is because, on these datasets, the gaps are \emph{nearly} identically distributed after removing the extreme values, thus better satisfying the requirement of using Theorem~\ref{theorem:segment_coverage}.

\noindent\textbf{Time Cost Model.}
For each pair of $(\epsilon_i, \epsilon_\ell)\in\mathcal{E}\times\mathcal{E}$, where $\mathcal{E}=\{2^j\mid j=2,3,\cdots,10\}$, we estimate the index lookup overhead as $Cost(\epsilon_i,\epsilon_\ell)$ using the time cost model (i.e., Eq.~\eqref{eq:time_cost_total}--Eq.~\eqref{eq:pgm_height_est}), and then physically construct the corresponding $(\epsilon_i, \epsilon_\ell)$-PGM-Index to measure the actual lookup time (averaged over a given workload). 

Figure~\ref{fig:time_est} visualizes the relationship between the true index lookup overhead and the cost model's estimation. 
The closer the points in Figure~\ref{fig:time_est} are to the line $y=x$, the more accurate the estimation. 
From the results, our cost model closely approximates the true index lookup overhead, especially for the three synthetic datasets \texttt{uniform}, \texttt{normal}, and \texttt{lognormal}. 
This is because synthetic datasets strictly follow i.i.d.~gaps assumptions, leading to more precise estimates of the index height (Eq.~\eqref{eq:pgm_height_est}), which significantly affects the total time cost estimation (Eq.~\eqref{eq:internal_cost_overall}). 

\noindent\textbf{Parameter Tuning Strategy.} 
We finally evaluate PGM++'s parameter tuner as introduced in Section~\ref{subsec:cost_model}. 
For a given $\epsilon_\ell$, which is directly solved given a pre-specified storage budget (Eq.~\eqref{eq:eps_ell_opt}), we record the index lookup overhead for PGM-Indexes with different $\epsilon_i$ configurations: (a) $T_{\text{PGM++}}$, where $\epsilon_i$ is automatically tuned using our cost model, (b) $T_{\text{rand}}$, where $\epsilon_i$ is randomly selected, and (c) $T_{\text{opt}}$, which is the optimal time cost by testing all possible $\epsilon_i$ values. 

Figure~\ref{fig:auto_config} reports the \emph{relative} index lookup overhead~w.r.t.~different $\epsilon_\ell$ settings (i.e., $T_{\text{PGM++}}/T_{\text{opt}}-1$ and $T_{\text{rand}}/T_{\text{opt}}-1$). 
From the results, across all datasets and $\epsilon_\ell$ settings (i.e., storage budgets), PGM++'s automatic parameter tuning strategy consistently finds a better $\epsilon_i$ to reduce the index lookup overhead. 
Specifically, in \textbf{46\%} of cases, PGM++ successfully picks the \textbf{optimal} $\epsilon_i$, and in \textbf{91\%} of cases, PGM++ finds a configuration that is only $\mathbf{<10\%}$ worse than the optimal one in terms of actual index lookup overhead. 

\noindent\textbf{Takeaways.} 
The experimental results reveal that in over \textbf{90\%} of cases, PGM++'s parameter tuner identifies a \textbf{near-optimal} index configuration, introducing less than \textbf{10\%} extra index lookup overhead. 
In addition, our parameter tuner is much more efficient than CDFShop~\cite{DBLP:conf/sigmod/MarcusZK20} designed for optimizing RMI structures (requiring <\SI{1}{\micro\second} v.s.~>10 minutes). 
This is because instead of physically constructing the index, our method only depends on gap distribution characteristics, which can be pre-computed and re-used.

% \begin{figure}[t]
%      \centering
%      \begin{subfigure}[b]{0.17\textwidth}
%          \centering
%          \includegraphics[width=\textwidth]{figures/fb_cost_model.pdf}
%          \caption{Dataset \texttt{fb}.}
%          \label{fig:fb_cost_model}
%      \end{subfigure}
%      \begin{subfigure}[b]{0.17\textwidth}
%          \centering
%          \includegraphics[width=\textwidth]{figures/wiki_cost_model.pdf}
%          \caption{Dataset \texttt{wiki}.}
%          \label{fig:wiki_cost_model}
%      \end{subfigure}

%      \begin{subfigure}[b]{0.17\textwidth}
%          \centering
%          \includegraphics[width=\textwidth]{figures/osm_cost_model.pdf}
%          \caption{Dataset \texttt{osm}.}
%          \label{fig:osm_cost_model}
%      \end{subfigure}
%      \begin{subfigure}[b]{0.17\textwidth}
%          \centering
%          \includegraphics[width=\textwidth]{figures/books_cost_model.pdf}
%          \caption{Dataset \texttt{books}.}
%          \label{fig:books_cost_model}
%      \end{subfigure}
%      \vspace{1pt}
%         \caption{Latency w.r.t.~data size for linear search, standard binary search (\texttt{std::lower\_bound}), and branchless binary search on platforms \texttt{arm} and \texttt{x86-1}.}
%         \label{fig:bench_search_results}
% \end{figure}

\begin{figure*}
    \centering
    \includegraphics[width=0.98\textwidth]{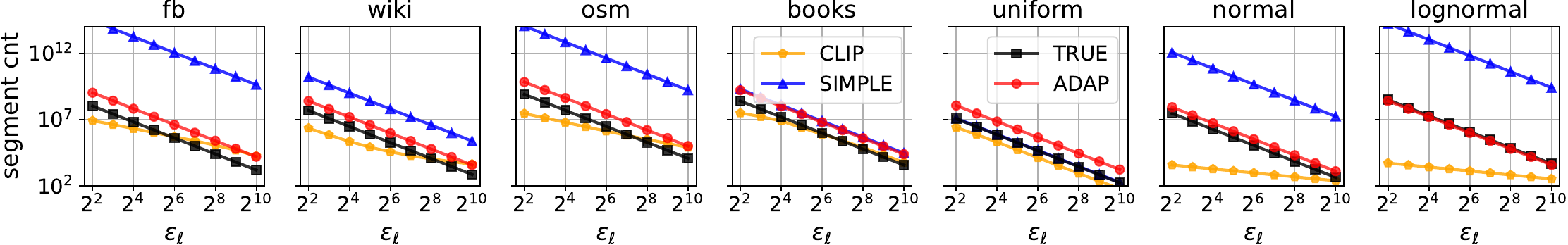}
    \caption{Evaluation of the space cost model (Eq.~\eqref{eq:space_cost}). For each $\epsilon_\ell$, we compare three leaf segment count estimators: (a) \textsf{SIMPLE}, (b) \textsf{CLIP}, and (c) \textsf{ADAP}. \textsf{TRUE} refers to the actual observed leaf segment count. }
    \label{fig:space_est}
    \vspace{-2ex}
\end{figure*}

\begin{figure*}
    \centering
    \includegraphics[width=0.98\textwidth]{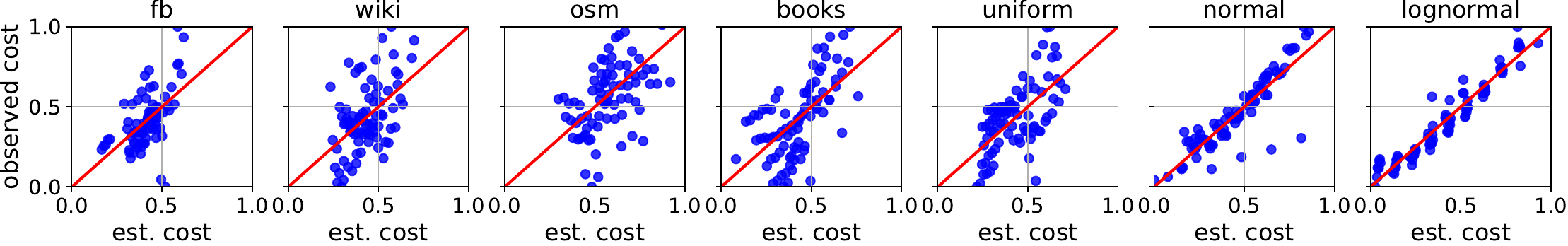}
    \caption{Evaluation of the time cost model (Eq.~\eqref{eq:time_cost_total}--Eq.~\eqref{eq:pgm_height_est}). We plot the true index lookup costs (normalized) against the estimated costs (normalized) on platform \texttt{x86-1}, where each point corresponds to a unique pair of $(\epsilon_i, \epsilon_\ell)$ configuration. }
    \label{fig:time_est}
    \vspace{-2ex}
\end{figure*}

\begin{figure*}
    \centering
    \includegraphics[width=0.98\textwidth]{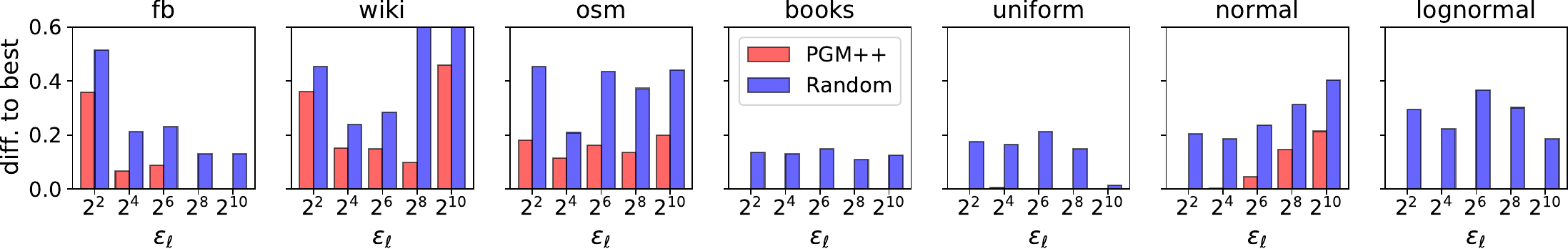}
    \caption{Evaluation of parameter tuning. The y-axis is the relative difference compared to the optimal configuration when fixing $\epsilon_\ell$. Red bars and blue bars refer to the $\epsilon_i$ settings selected by using PGM++'s cost model and randomly picking, respectively.}
    \label{fig:auto_config}
    \vspace{-2ex}
\end{figure*}

% \subsection{Compared Methods}

% \subsection{Pareto Analysis}

% \subsection{Cache Efficiency}

% \subsection{Construction Efficiency}

%% file: sections/7-related-work.tex
\section{Related Work}\label{sec:related_works}
\noindent\textbf{Learned Indexes.} 
Indexing one-dimensional sorted keys has been a well-explored topic for decades. 
While mainstream tree-based indexes (e.g., B+-tree~\cite{DBLP:journals/csur/Comer79}, FAST~\cite{DBLP:conf/sigmod/KimCSSNKLBD10}, ART~\cite{DBLP:conf/sigmod/BinnaZPSL18}, Wormhole~\cite{DBLP:conf/eurosys/WuNJ19}, HOT~\cite{DBLP:conf/sigmod/BinnaZPSL18}, etc.) are widely adopted in commercial DBMS, a new class of data structure, known as \emph{learned index}, has recently gained significant attention in both academia and industry~\cite{DBLP:conf/sigmod/KraskaBCDP18, DBLP:journals/pvldb/FerraginaV20,DBLP:conf/sigmod/DingMYWDLZCGKLK20,DBLP:journals/pvldb/WuZCCWX21,DBLP:journals/pvldb/ZhangG22,zhang2024making,DBLP:journals/pvldb/0006CJLXWLWZZWR22,DBLP:journals/pvldb/WuCYSKX22,DBLP:journals/tos/WangCWTW22,DBLP:conf/ppopp/TangWDHWWC20}.
Intuitively, learned indexes directly fit the CDF over sorted keys with controllable error to perform an error-bounded last-mile search. 
By properly organizing the model structure, learned indexes offer the potential for superior space-time trade-offs compared to conventional tree-based indexes~\cite{DBLP:journals/pvldb/MarcusKRSMK0K20,DBLP:journals/pvldb/WongkhamLLZLW22}. 

Existing learned indexes can be roughly categorized as either RMI-like~\cite{DBLP:conf/sigmod/KraskaBCDP18} or PGM-like~\cite{DBLP:journals/pvldb/FerraginaV20}, based on whether the error-bounded search occurs during the index traversal phase. 
This work delves deeply into the theoretical and empirical aspects of the PGM-Index, highlighting its potential to be \emph{practically} embedded into real DBMS. 

\noindent\textbf{Learned Index Theories.} 
Unlike tree-based indexes, which are supported by well-established theoretical foundations, the effectiveness of learned indexes has largely been demonstrated through \emph{empirical results}. 
Ferragina et~al.~\cite{DBLP:journals/pvldb/FerraginaV20,DBLP:conf/icml/FerraginaLV20} first prove that the expected time and space complexities of a PGM-Index with error constraint $\epsilon$ on $N$ keys should be $O(\log N)$ and $O(N/\epsilon^2)$, respectively. 
In parallel, another recent work~\cite{DBLP:conf/icml/ZeighamiS23} focuses on an RMI variant with \emph{piece-wise constant} models, achieving an index lookup time of $O(\log\log N)$ but using \emph{super-linear} space of $O(N\log N)$. 

In this work, we tighten the results of~\cite{DBLP:conf/icml/FerraginaLV20} by achieving a sub-logarithmic time complexity of $O(\log\log N)$ while maintaining \emph{linear} space, $O(N/\epsilon^2)$, for PGM-Indexes. 
To the best of our knowledge, this is the tightest bound among all existing learned indexes. 

\noindent\textbf{Learned Index Cost Model.} 
Modeling the space and time overheads of an index structure is crucial for both index parameter configuration and DBMS query optimization. 
Existing learned indexes mainly adopt a workload-based cost model, which assumes prior knowledge of the query distribution~\cite{DBLP:journals/pvldb/ZhangG22, DBLP:conf/sigmod/MarcusZK20}. 
In contrast, by extending the theoretical results, we establish a cost model for PGM-like indexes without \emph{any} assumptions on query workloads. 
As our cost model is simple, parameter tuning based on it is much more efficient than workload-driven approaches, making it more feasible to be integrated into practical DBMS. 
 
\noindent\textbf{AI4DB.} 
Beyond learned indexing, recent advancements in AI are reshaping traditional approaches on decades-old data management challenges, such as query planning~\cite{DBLP:journals/pvldb/ZhuCDCPWZ23,DBLP:conf/sigmod/MarcusNMTAK21,DBLP:journals/pvldb/YuC0L22}, cardinality estimation~\cite{DBLP:conf/cidr/KipfKRLBK19,DBLP:journals/pvldb/WangQWWZ21}, approximate query processing~\cite{DBLP:conf/cidr/MaSAKT21,DBLP:conf/icde/Thirumuruganathan20}, SQL generation~\cite{DBLP:conf/sigmod/WeirUGCIRBGHEcB20,DBLP:journals/pvldb/KimSHL20}, DBMS configuration~\cite{DBLP:conf/sigmod/AkenPGZ17,DBLP:conf/sigmod/ZhangLZLXCXWCLR19}, etc.

%% file: sections/8-conclusion.tex
\vspace{-1ex}
\section{Conclusion and Future Work}\label{sec:conclusion}
This work provides a thorough theoretical and experimental revisit to the PGM-Index. 
We establish a new bound for the PGM-Index by showing the $O(\log\log N)$ index lookup time while using $O(N/G)$ space. 
We further identify that costly internal error-bounded search operations have become a bottleneck in practice. 
Based on such findings, we propose PGM++, a simple yet effective PGM-Index variant, by improving the internal search subroutine and configuring index hyper-parameters based on accurate cost models. 
Extensive experimental results demonstrate that PGM++ speeds up index lookup queries by up to $\mathbf{2.31\times}$ and $\mathbf{1.56\times}$ compared to the original PGM-Index and the optimized RMI implementation, respectively.

\noindent\textbf{Future Work.}  
\ding{182} Our theoretical results inherit the i.i.d.~assumption on \emph{gaps} from previous analyses. 
In our future work, we aim to relax this assumption to demonstrate that the sub-logarithmic bound still holds for weakly correlated data. 
\ding{183} To further accelerate PGM++, we plan to fully exploit architecture-aware optimizations like memory pre-fetching and SIMD. 
Additionally, we will release a GPU-accelerated version of PGM++. 